\documentclass[12pt]{article}
\usepackage[a4paper,margin=30mm]{geometry}
\usepackage[utf8]{inputenc}
\usepackage{amsmath}
\usepackage{amsthm}
\usepackage{amsfonts}
\usepackage{amssymb}
\usepackage{latexsym}
\usepackage{amscd}
\usepackage{stmaryrd}
\usepackage{amsbsy}

\usepackage[hidelinks]{hyperref}
\allowdisplaybreaks

\usepackage{fourier}

\newcommand{\Vir}{\mathrm{Vir}}
\newcommand{\Z}{\mathbb{Z}}
\newcommand{\C}{\mathbb{C}}

\theoremstyle{definition}
\newtheorem{thm}{Theorem}[section]
\newtheorem{dfn}[thm]{Definition}
\newtheorem{exmp}[thm]{Example}

\newtheorem{cor}[thm]{Corollary}
\newtheorem{lem}[thm]{Lemma}
\newtheorem{prop}[thm]{Proposition}

\theoremstyle{remark}
\newtheorem{re}[thm]{Remark}

\newtheorem*{claim}{Claim}
\numberwithin{equation}{section}

\begin{document}
\title{Neveu--Schwarz Irregular Vertex Operators, Decomposition Theorems, and Bilinear Operators}
\author{Hajime Nagoya \\ School of Mathematics and Physics, Kanazawa University, \\Kanazawa, Ishikawa 920-1192, Japan \\ E-mail:
\href{mailto:nagoya@se.kanazawa-u.ac.jp}{nagoya@se.kanazawa-u.ac.jp}}
\date{}
\maketitle
\begin{abstract}

We construct rank-zero irregular vertex operators for the
Neveu--Schwarz algebra as linear maps between irregular Verma modules
of the same rank satisfying the usual superconformal commutation
relations and an irregular asymptotic condition. Under the
nondegeneracy assumption \(\Lambda_{2p}\neq0\), we prove their existence
and uniqueness. We then extend the decomposition theorem for the
Neveu--Schwarz algebra to the irregular setting: the tensor product of
the free-fermion Fock module with a rank-\(p\) Neveu--Schwarz irregular
Verma module decomposes into an infinite direct sum of tensor products
of two Virasoro irregular Verma modules. We also prove a compatible
decomposition of the irregular vertex operators into tensor products
of Virasoro irregular vertex operators. Using suitable pairings and mode insertions, we derive bilinear differential equations for weighted sums of products of Virasoro irregular conformal blocks of types \((0,0,1)\) and \((0,2)\). After
explicit parameter identifications and a gauge transformation in the
Painlev\'e V case, the resulting bilinear differential operators agree
with those appearing in the quantum Painlev\'e V and IV tau-function
equations.

\end{abstract}

\section{Introduction}

The relation between two-dimensional conformal field theory \cite{BPZ} and integrable
systems has become a central theme in mathematical physics, especially since
the discovery of the AGT correspondence \cite{AGT}.  One striking manifestation is the relation between Virasoro conformal blocks
and Painlev\'e tau functions.  In 2012, Gamayun, Iorgov, and Lisovyy conjectured that the
tau function of the sixth Painlev\'e equation is represented by a Fourier-type
series of four-point Virasoro conformal blocks with central charge \(c=1\)
\cite{GIL}.  This correspondence was subsequently extended to the 
Painlev\'e equations \(\mathrm{P}_{\mathrm{V}}\) and
\(\mathrm{P}_{\mathrm{III}}\) \cite{GIL1}, and conjectured formulas were proved in
\cite{BS,GL,ILT,LNR18}.  The proof of \cite{BS} uses the decomposition of a Neveu--Schwarz
vertex operator into an
infinite sum of tensor products of two Virasoro vertex operators.

Irregular Virasoro vertex operators and irregular conformal blocks were
introduced in \cite{Nagoya2015,Nagoya2018} in order to formulate analogous
series expansions of the tau functions of
\(\mathrm{P}_{\mathrm{V}}\), \(\mathrm{P}_{\mathrm{IV}}\),
\(\mathrm{P}_{\mathrm{III}}\), and \(\mathrm{P}_{\mathrm{II}}\) at irregular
singular points.  For \(\mathrm{P}_{\mathrm{V}}\) and
\(\mathrm{P}_{\mathrm{IV}}\), these expansions were proved by taking
degeneration limits of regular Virasoro vertex operators \cite{NN}.
Connections between Painlev\'e equations and gauge theory were also studied
in \cite{BLMST}, and this framework has recently been extended to quantum
Painlev\'e tau functions and their bilinear equations \cite{BST}.

The regular Neveu--Schwarz theory underlying these results has a remarkable
decomposition structure.  Belavin et al.\ proved that the tensor product of
the free-fermion Fock module with a Verma module for the Neveu--Schwarz
algebra decomposes into an infinite direct sum of tensor products of two
Virasoro Verma modules \cite{BBFLT}.  Together with the corresponding
decomposition of regular vertex operators, this result provides a
representation-theoretic origin of bilinear relations for Virasoro conformal
blocks.  It is therefore natural to seek an irregular counterpart of this
structure.  Such an extension is not merely formal: the modules involved are
higher-order Whittaker modules, the action of an irregular vertex operator on
the generating vector contains an essential singular factor, and the
decomposition must be compatible with the irregular eigenvalues.

The first main result of this paper is the construction of rank-zero
irregular vertex operators for the Neveu--Schwarz algebra.  For
\(p\in\mathbb Z_{\geq1}\), let
\(M_{\mathrm{NS}}^{\Lambda,[p]}\) be the rank-$p$ irregular Verma module 
generated by an irregular vector \(|\Lambda\rangle\), where
\[
 \Lambda=(\Lambda_p,\Lambda_{p+1},\ldots,\Lambda_{2p})
 \in\mathbb C^{p+1}.
\]
These modules are higher-order Whittaker modules, and the nondegeneracy
condition
\[
 \Lambda_{2p}\neq0
\]
implies their irreducibility by the criterion of \cite{LPX}.  We define a
pair of irregular vertex operators
\[
 \Phi_{\Lambda',\Lambda}^{\Delta}(z),\quad
 \Psi_{\Lambda',\Lambda}^{\Delta}(z):
 M_{\mathrm{NS}}^{\Lambda,[p]}
 \longrightarrow
 M_{\mathrm{NS}}^{\Lambda',[p]}.
\]
Their commutation relations with the Neveu--Schwarz generators are the same
as those of ordinary Neveu--Schwarz vertex operators, whereas their action
on \(|\Lambda\rangle\) contains an irregular factor
\[
 z^\alpha
 \exp\left(\sum_{i=1}^{p}\frac{\beta_i}{z^i}\right).
\]
Since the source and target modules have the same irregular rank, these
operators may be regarded as rank-zero irregular Neveu--Schwarz vertex
operators.  For a fixed central charge, Theorem~\ref{thm_IVO} proves that, once
\(\Lambda\), \(\Delta\), and \(\beta_p\) are given, the operators exist
and are unique.  In particular, the target weight is determined by
\[
 \Lambda'_p=\Lambda_p-p\beta_p,
 \qquad
 \Lambda'_n=\Lambda_n
 \quad (p+1\leq n\leq2p).
\]

The second main result, and the structural core of the paper, is an irregular
analog of the decomposition theorem of \cite{BBFLT}.  Using irregular
free-field realizations of the Neveu--Schwarz algebra, we embed
\(\mathrm{Vir}\oplus\mathrm{Vir}\) into
\(\mathrm{F}\oplus\mathrm{NS}\). Under the assumptions
\[
 P_p\neq0,\qquad b\neq0,\qquad b^2\neq1,
\]
Theorem~\ref{thm_decomposition_module} gives
\[
 M_{\mathrm{F}\oplus\mathrm{NS}}^{\Lambda,[p]}
 \cong
 \bigoplus_{2m\in\mathbb Z}
 M_{\mathrm{Vir}\oplus\mathrm{Vir}}^{\Lambda,m,[p]}.
\]
Each summand on the right-hand side is the tensor product of two irregular
Virasoro Verma modules, and its generating irregular vector is denoted by
\[
 |P,m\rangle
 =
 |\Lambda^{(m,1)}\rangle
 \otimes
 |\Lambda^{(m,2)}\rangle.
\]
The proof uses a degree filtration, the corresponding associated graded
module, and a graded-character identity obtained from the Jacobi triple
product formula.

The module decomposition is accompanied by a decomposition theorem for
irregular vertex operators.  Theorem~\ref{thm_IVO_decomposition} shows that
the restriction of
\(1\otimes\Phi_{\Lambda',\Lambda}^{\Delta}(z)\) from the \(m\)-th source
summand to the \(n\)-th target summand is a scalar multiple of the tensor
product of two Virasoro irregular vertex operators:
\[
 \left.
 1\otimes\Phi_{\Lambda',\Lambda}^{\Delta}(z)
 \right|_{
 M_{\mathrm{Vir}\oplus\mathrm{Vir}}^{\Lambda,m,[p]}
 \to
 M_{\mathrm{Vir}\oplus\mathrm{Vir}}^{\Lambda',n,[p]}
 }
 =
 \mathsf a_{mn}\,
 V_{{\Lambda'}^{(n,1)},\Lambda^{(m,1)}}^{\Delta^{(1)}}(z)
 \otimes
 V_{{\Lambda'}^{(n,2)},\Lambda^{(m,2)}}^{\Delta^{(2)}}(z).
\]
Thus, both the module structure and the vertex-operator structure of the
regular decomposition theorem extend to the irregular setting.

As an application, we introduce suitable pairings on
\(\mathrm{F}\oplus\mathrm{NS}\) modules and apply the decomposition theorem
to the corresponding matrix elements.  Insertions of the operators \(H_n\)
then produce bilinear differential equations for weighted sums of products of
Virasoro irregular conformal blocks of types \((0,0,1)\) and \((0,2)\).
We prove that the resulting bilinear differential operators coincide, after
explicit parameter identifications and, in the
\(\mathrm{P}_{\mathrm{V}}\) case, a gauge transformation, with the bilinear
operators appearing in the quantum
\(\mathrm{P}_{\mathrm{V}}\) and \(\mathrm{P}_{\mathrm{IV}}\) tau-function
equations of \cite{BST}.

The result established here is an identification at the level of bilinear
differential operators.  A stronger statement would identify the quantum
Painlev\'e tau functions themselves with Zak transforms of Virasoro
irregular conformal blocks.  Such an identification requires explicit
formulas for the scalar coefficients in the conformal-block expansions and
is not pursued in the present paper.  The relation between these
coefficients, pairings of embedded Virasoro irregular vectors, and
degeneration limits of regular vertex-operator decompositions is discussed
in the final section.

The paper is organized as follows.  In Section~2, we define irregular Verma
modules for the Neveu--Schwarz algebra, introduce their degree filtration,
and recall their irreducibility under the condition
\(\Lambda_{2p}\neq0\).  In Section~3, we define rank-zero irregular
Neveu--Schwarz vertex operators and prove their existence and uniqueness.
In Section~4, we construct the irregular free-field realizations, prove the
decomposition theorems for modules and vertex operators, and discuss the
decomposition of the free-fermion field.  In Section~5, we introduce the
pairings and \(H_n\)-insertions, derive the bilinear equations for irregular
conformal blocks, and compare the resulting operators with the quantum
Painlev\'e bilinear operators.  Section~6 contains a discussion of the
stronger Zak-transform formulas and possible extensions of the present
framework.

\section{Irregular Verma module}

The Neveu--Schwarz algebra 
\begin{equation*}
{\bf \mathrm{NS}}=\bigoplus_{n\in\Z}\C L_n\oplus 
\bigoplus_{r\in \Z+\frac{1}{2}}
\C G_r
\oplus \C c 
\end{equation*}
 is a Lie superalgebra with (anti)commutation relations
\begin{align*}
&[L_m,L_n]=(m-n)L_{m+n}+\frac{c}{8}(m^3-m)\delta_{m+n,0},
\\
&[L_m,G_r]=\left( \frac{m}{2}-r\right)G_{m+r},
\\
&[G_r,G_s]_+=2L_{r+s}+\frac{c}{2}\left(r^2-\frac{1}{4}
\right)\delta_{r+s,0},
\\
&[\mathrm{NS},c]=0.
\end{align*}

Fix $p\in\Z_{\geq 1}$. Set 
\begin{equation*}
	\mathrm{NS}_{p}=\bigoplus_{n\geq p}\C L_n \oplus \bigoplus_{r>p}\C G_r\oplus \C c. 
\end{equation*}
For $\Lambda=(\Lambda_p,\ldots,\Lambda_{2p})\in \C^{p+1}$, 
let $\C|\Lambda\rangle $ be one-dimensional $\mathrm{NS}_p$-module with 
\begin{align*}
	&L_n|\Lambda\rangle=\Lambda_n|\Lambda\rangle\quad (n\ge p),
	\quad G_r|\Lambda\rangle=0\quad \left(r>p\right), 
	\end{align*}
	where $\Lambda_n=0$ if $n>2p$. The element $c$ acts as 
	multiplication by a complex number. 
Define an induced module 
\begin{equation*}
	M^{\Lambda,[p]}_{\mathrm{NS}}=\mathrm{Ind}^{\mathrm{NS}}_{\mathrm{NS}_p}\C |\Lambda\rangle. 
\end{equation*}	
We call this module  
an irregular Verma module of the Neveu--Schwarz algebra of rank $p$ and $|\Lambda\rangle$ an irregular vector.  
An irregular Verma module is a higher-order Whittaker module considered in \cite{LPX}.

For a partition $\lambda=(\lambda_1,\ldots,\lambda_\ell)$, let 
\begin{align*}
F_{-\lambda}=F_{-\lambda_1}\cdots F_{-\lambda_\ell}, 
\end{align*}
where
\begin{align*}
F_{-n}=\begin{cases}
L_{p-n/2} & n \text{ even}, 
\\[2mm]
G_{p-n/2} & n \text{ odd},
\end{cases}\qquad n\in\mathbb Z_{>0}.
\end{align*}
The vectors $F_{-\lambda}|\Lambda\rangle$ for which the multiplicity of
each odd part is at most one form a basis of the irregular Verma module
$M^{\Lambda,[p]}_{\mathrm{NS}}$.  

We set $\deg(F_{-\lambda}|\Lambda\rangle)=|\lambda|/2$ and    
\[
 U_{\leq d}
 =
 \operatorname{span}_{\mathbb C}
 \left\{
 F_{-\lambda}|\Lambda\rangle
 \mathrel{}\middle|\mathrel{}
 \frac{|\lambda|}{2}\leq d
 \right\},
 \qquad
 d\in\frac12\mathbb Z,
\]
and put \(U_{\leq d}=0\) for \(d<0\).  This defines an increasing
filtration of \(M_{\mathrm{NS}}^{\Lambda,[p]}\).

Set
\[
 \widetilde F_n
 =
 \begin{cases}
  L_{p+n/2}-\Lambda_{p+n/2},
   & n \text{ even},\\[2mm]
  G_{p+n/2},
   & n \text{ odd},
 \end{cases}
 \qquad n\in\mathbb Z_{>0},
\]
where \(\Lambda_k=0\) for \(k>2p\).

\begin{lem}[Filtration lemma]\label{lem:NS-filtration}
For every \(n\in\mathbb Z_{>0}\) and
\(d\in\frac12\mathbb Z\),
\begin{equation*}
 \widetilde F_n U_{\leq d}
 \subset
 U_{\leq d-n/2}.
\end{equation*}

Moreover, under the PBW identification
\[
 \operatorname{gr}M_{\mathrm{NS}}^{\Lambda,[p]}
 \simeq
 \mathbb C[x_2,x_4,\ldots]
 \otimes
 \bigwedge(\xi_1,\xi_3,\ldots),
\]
where \(x_n\) and \(\xi_n\) are the symbols of
\(F_{-n}\) for even and odd \(n\), respectively, the induced
homogeneous operator on the associated graded module is
\begin{equation*}
 \operatorname{gr}(\widetilde F_n)
 =
 \begin{cases}
  n\Lambda_{2p}\dfrac{\partial}{\partial x_n},
   & n \text{ even},\\[3mm]
  2\Lambda_{2p}\dfrac{\partial}{\partial \xi_n},
   & n \text{ odd}.
 \end{cases}
\end{equation*}
Here \(\partial/\partial\xi_n\) denotes the left superderivative.
\end{lem}

\begin{proof}
It is enough to consider a basis vector
\(F_{-\lambda}|\Lambda\rangle\) with
\(|\lambda|/2\leq d\).  Rewrite
\(\widetilde F_nF_{-\lambda}\) in PBW order.  The Neveu--Schwarz
commutation relations are homogeneous with respect to the usual
\(L_0\)-degree.  Assign filtration degree \(0\) to the factors
belonging to \(\mathrm{NS}_p\).  Then the remaining part of every PBW
monomial is of the form \(F_{-\mu}\), with
\[
 \frac{|\mu|}{2}
 \leq
 \frac{|\lambda|}{2}-\frac n2.
\]
After acting on \(|\Lambda\rangle\), the factors belonging to
\(\mathrm{NS}_p\) act by scalars or zero.  Moreover, the term in which
\(\widetilde F_n\) is moved all the way to the right vanishes, since
\(\widetilde F_n|\Lambda\rangle=0\).  Hence
\[
 \widetilde F_nU_{\leq d}
 \subset
 U_{\leq d-n/2}.
\]

The equality in the filtration estimate can occur only from the supercommutator of \(\widetilde F_n\) with a factor \(F_{-n}\).  For even
\(n\),
\[
 [L_{p+n/2},L_{p-n/2}]
 =
 nL_{2p},
\]
whereas for odd \(n\),
\[
 \{G_{p+n/2},G_{p-n/2}\}
 =
 2L_{2p}.
\]
Since \(L_{2p}|\Lambda\rangle=\Lambda_{2p}|\Lambda\rangle\), the
principal symbol is
\[
 \operatorname{gr}(\widetilde F_n)
 =
 \begin{cases}
  n\Lambda_{2p}\dfrac{\partial}{\partial x_n},
   & n\text{ even},\\[3mm]
  2\Lambda_{2p}\dfrac{\partial}{\partial\xi_n},
   & n\text{ odd}.
 \end{cases}
\]
\end{proof}

\begin{cor}\label{cor_n_kill}
Assume that \(\Lambda_{2p}\neq0\).  If
\(\widetilde F_nu=0\) for every \(n\in\mathbb Z_{>0}\), then
\(u\in U_{\leq0}\).
\end{cor}

\begin{proof}
Suppose that \(u\notin U_{\leq0}\), and let \(d>0\) be its highest
filtration degree.  Denote its nonzero leading symbol by
\[
 \overline u\in
 \operatorname{gr}_dM_{\mathrm{NS}}^{\Lambda,[p]}.
\]
By Lemma~\ref{lem:NS-filtration} and the assumption
\(\widetilde F_nu=0\), the element \(\overline u\) is annihilated by
all the derivatives
\[
 \frac{\partial}{\partial x_n}
 \quad(n\text{ even}),
 \qquad
 \frac{\partial}{\partial\xi_n}
 \quad(n\text{ odd}).
\]
Hence \(\overline u\) is constant.  This contradicts \(d>0\).
Therefore \(u\in U_{\leq0}\).
\end{proof}

  For $u\in M^{\Lambda,[p]}_{\mathrm{NS}}$, denote the constant term of $u$ by $\{u\}$. 
For a partition $\lambda=(\lambda_1,\ldots,\lambda_\ell)$, let 
\begin{equation*}
\widetilde F_{\lambda}=\widetilde F_{\lambda_\ell}\cdots \widetilde F_{\lambda_1}, 
\end{equation*}
  
\begin{lem}\label{lem_pairing_LL}
For partitions $\lambda$ and $\mu$ 
such that $ |\lambda|\ge |\mu|$, 
\begin{equation*}
\left\{\widetilde{F}_\lambda
F_{-\mu}|\Lambda\rangle\right\}=\left\{\begin{matrix}
0, & \lambda\neq \mu,
\\
\left(2\Lambda_{2p}\right)^{\ell(\lambda)}\displaystyle
\prod_{\substack{1\le i\le m\\i\text{ even}}}\left(\frac{i}{2}\right)^{k_i}k_i!\prod_{\substack{1\le i\le m \\ i\text{ odd}}}k_i!, & 
\lambda=\mu,
\end{matrix}\right.
\end{equation*}
where $\lambda=(m^{k_m},(m-1)^{k_{m-1}},
\ldots, 2^{k_2},1^{k_1})$ ($k_i\in\Z_{\ge 0}$). 
\end{lem}
\begin{proof}
Let \(S_\mu\) denote the leading symbol of
\(F_{-\mu}|\Lambda\rangle\) in
\(\operatorname{gr}M_{\mathrm{NS}}^{\Lambda,[p]}\).
Repeated application of Lemma~\ref{lem:NS-filtration} gives
\[
 \widetilde F_\lambda
 F_{-\mu}|\Lambda\rangle
 \in
 U_{\leq (|\mu|-|\lambda|)/2}.
\]
Hence, if \(|\lambda|>|\mu|\), then
\[
 \widetilde F_\lambda F_{-\mu}|\Lambda\rangle=0,
\]
and in particular its constant term vanishes.

Suppose that \(|\lambda|=|\mu|\).  Since
\(U_{\leq0}=\mathbb C|\Lambda\rangle\) and
\(U_{\leq-1/2}=0\), the constant term is determined exactly by the
degree-zero principal symbol.  Thus, after identifying
\(\operatorname{gr}_0M_{\mathrm{NS}}^{\Lambda,[p]}\) with
\(\mathbb C\), we have
\[
 \left\{
 \widetilde F_\lambda F_{-\mu}|\Lambda\rangle
 \right\}
 =
 D_{\lambda_\ell}\cdots D_{\lambda_1}S_\mu,
\]
where
\[
 D_n=
 \begin{cases}
  n\Lambda_{2p}\dfrac{\partial}{\partial x_n},
    & n\text{ even},\\[3mm]
  2\Lambda_{2p}\dfrac{\partial}{\partial\xi_n},
    & n\text{ odd}.
 \end{cases}
\]

Write
\[
 \lambda=(m^{k_m},\ldots,2^{k_2},1^{k_1}),
 \qquad
 \mu=(m^{h_m},\ldots,2^{h_2},1^{h_1}),
\]
where \(k_i,h_i\in\mathbb Z_{\geq0}\), and where
\(k_i,h_i\leq1\) for odd \(i\).  The differential monomial
\(D_{\lambda_\ell}\cdots D_{\lambda_1}\) annihilates
\(S_\mu\) unless
\[
 k_i\leq h_i
 \qquad\text{for every }i.
\]
Since \(|\lambda|=|\mu|\), these inequalities imply
\(k_i=h_i\) for every \(i\).  Hence the constant term vanishes unless
\(\lambda=\mu\).

It remains to compute the diagonal term.  For even \(i\),
\[
 \left(
 i\Lambda_{2p}\frac{\partial}{\partial x_i}
 \right)^{k_i}
 x_i^{k_i}
 =
 \left(i\Lambda_{2p}\right)^{k_i}k_i!,
\]
whereas for odd \(i\),
\[
 \left(
 2\Lambda_{2p}\frac{\partial}{\partial\xi_i}
 \right)^{k_i}
 \xi_i^{k_i}
 =
 \left(2\Lambda_{2p}\right)^{k_i}k_i!.
\]
Here the reverse ordering in the definition of
\(\widetilde F_\lambda\) ensures that the successive left
superderivatives remove the odd generators in their PBW order, so no
additional sign occurs.  Therefore,
\begin{align*}
 \left\{
 \widetilde F_\lambda F_{-\lambda}|\Lambda\rangle
 \right\}
 &=
 \prod_{\substack{1\leq i\leq m\\ i\text{ even}}}
 \left(i\Lambda_{2p}\right)^{k_i}k_i!
 \prod_{\substack{1\leq i\leq m\\ i\text{ odd}}}
 \left(2\Lambda_{2p}\right)^{k_i}k_i!
 \\
 &=
 \left(2\Lambda_{2p}\right)^{\ell(\lambda)}
 \prod_{\substack{1\leq i\leq m\\ i\text{ even}}}
 \left(\frac{i}{2}\right)^{k_i}k_i!
 \prod_{\substack{1\leq i\leq m\\ i\text{ odd}}}
 k_i!,
\end{align*}
as required.
\end{proof}

For \(0\le n\le m\), let
\[
\mathcal G_{n,m}
=
\left(
 \left\{
  \widetilde F_\lambda F_{-\mu}|\Lambda\rangle
 \right\}
\right)_{n\le|\lambda|,|\mu|\le m},
\]
where the rows and columns are ordered by increasing partition size.

For example, 
\begin{align*}
\mathcal G_{1,2}
=&\begin{pmatrix}
\left\{ G_{1/2+p}
G_{-1/2+p}|\Lambda\rangle\right\}&\left\{G_{1/2+p}
L_{-1+p}|\Lambda\rangle\right\}
\\
\left\{\left(L_{1+p}-\Lambda_{1+p}\right)
G_{-1/2+p}|\Lambda\rangle\right\}&
\left\{\left(L_{1+p}-\Lambda_{1+p}\right)
L_{-1+p}|\Lambda\rangle\right\}
\end{pmatrix}
\\
=&\begin{pmatrix}
2\Lambda_{2p}&0
\\
0&
2\Lambda_{2p}
\end{pmatrix}. 
\end{align*}

\begin{cor}\label{cor_det}
For any $0\leq n\leq m$,
\begin{equation*}
\det \mathcal G_{n,m}
=
C\Lambda_{2p}^{N_{n,m}},
\end{equation*}
where $C$ is a positive integer and $N_{n,m}$ is the sum of the lengths
of all NS partitions whose sizes lie between $n$ and $m$.
\end{cor}
\begin{proof}
Lemma~\ref{lem_pairing_LL} implies that $\mathcal G_{n,m}$ is upper triangular.
Its determinant is therefore the product of the diagonal entries computed
in Lemma~\ref{lem_pairing_LL}.
\end{proof}

The following irreducibility statement is known from the theory of higher-order
Whittaker modules \cite{LPX}.  We recall it here because it is used in the
construction of irregular vertex operators.
\begin{prop}\label{prop:NS-Verma-irreducible}
If $\Lambda_{2p}\neq 0$, then $M_{\mathrm{NS}}^{\Lambda,[p]}$ is irreducible.
\end{prop}
\begin{proof}
    Let \(W\) be a nonzero submodule and choose a nonzero vector \(u\in W\) of minimal filtration degree. Since \(\widetilde F_n\) lowers the filtration degree for every \(n>0\), the minimality of \(u\) implies \(\widetilde F_nu=0\) for all \(n>0\). Corollary~\ref{cor_n_kill} then gives \(u\in U_{\leq 0}=\mathbb C|\Lambda\rangle\). Hence \(|\Lambda\rangle\in W\), and therefore \(W=M_{\mathrm{NS}}^{\Lambda,[p]}\).
\end{proof}

\section{Irregular vertex operator}
In this section, we give a definition of irregular vertex operators 
of the Neveu--Schwarz algebra and prove that they exist and are unique if the irregular 
Verma module is irreducible.

\begin{dfn}
We define irregular vertex operators $\Phi_{\Lambda', \Lambda}^\Delta(z)$ 
and $\Psi_{\Lambda', \Lambda}^\Delta(z): M^{\Lambda,[p]}_{\mathrm{NS}}\to M^{\Lambda',[p]}_{\mathrm{NS}}$ by
\begin{align}
&\left[ L_n, \Phi_{\Lambda', \Lambda}^\Delta(z)\right]
=z^n\left( z\frac{\partial}{\partial z}+(n+1)\Delta\right)\Phi_{\Lambda', \Lambda}^\Delta(z), 
\label{eq:L_phi}
\\
&\left[ L_n, \Psi_{\Lambda', \Lambda}^\Delta(z)\right]
=z^n\left( z\frac{\partial}{\partial z}+(n+1)\left(\Delta+\frac{1}{2}\right)\right)\Psi_{\Lambda', \Lambda}^\Delta(z), 
\label{eq:L_psi}
\\
&\left[ G_r, \Phi_{\Lambda', \Lambda}^\Delta(z)\right]
=z^{r+\frac{1}{2}}\Psi_{\Lambda', \Lambda}^\Delta(z), 
\label{eq:G_phi}
\\
&\left[ G_r, \Psi_{\Lambda', \Lambda}^\Delta(z)\right]_+
=z^{r-\frac{1}{2}}\left( z\frac{\partial}{\partial z}+(2r+1)\Delta\right)\Phi_{\Lambda', \Lambda}^\Delta(z), 
\label{eq:G_psi}
\\
&\Phi_{\Lambda', \Lambda}^\Delta(z)|\Lambda\rangle=z^\alpha \exp\left(\sum_{i=1}^p\frac{\beta_i}{z^i}\right)\sum_{m=0}^\infty v_mz^{\frac{m}{2}},
\label{eq:phi_irr}
\\
&\Psi_{\Lambda', \Lambda}^\Delta(z)|\Lambda\rangle=z^{\alpha-\frac{p+1}{2}} \exp\left(\sum_{i=1}^p\frac{\beta_i}{z^i}\right)\sum_{m=0}^\infty u_mz^{\frac{m}{2}}
\label{eq:psi_irr}
\end{align}
for $n\in\Z$, $r\in\Z+1/2$, 
where $v_0=u_0=|\Lambda'\rangle$, $v_m,u_m\in M^{\Lambda',[p]}_{\mathrm{NS}}$. 
\end{dfn}

The commutation relations \eqref{eq:L_phi}--\eqref{eq:G_psi} between the generators 
and the irregular vertex operators are 
the same as those for the regular vertex operators.  
The actions of the irregular vertex operators on the irregular vector 
\eqref{eq:phi_irr} and \eqref{eq:psi_irr} are of irregular type if $p>0$. 
If $p=0$, they are of regular type, and if the 
Verma module $M^{\Lambda',[0]}_{\mathrm{NS}}$ is irreducible, then the regular vertex 
operators exist uniquely. When $p>0$, we have the following theorem. 

\begin{thm}\label{thm_IVO}
	 For a fixed central charge, if $p>0$ and $\Lambda_{2p}\neq 0$, then the irregular vertex operators 
	$\Phi_{\Lambda', \Lambda}^\Delta(z)$ and 
	$\Psi_{\Lambda', \Lambda}^\Delta(z): M^{\Lambda,[p]}_{\mathrm{NS}}\to M^{\Lambda',[p]}_{\mathrm{NS}}$ 
	exist and are uniquely determined by the given parameters 
	$\Lambda$, $\Delta$, $\beta_p$. More precisely, 
	\begin{equation*}
		\Lambda'_p=\Lambda_p-p\beta_p,\quad \Lambda'_n=\Lambda_n\quad (n=p+1,\ldots,2p),
	\end{equation*}
	$\alpha$, $\beta_i$ for $i=1,\ldots,p-1$ and $c_\lambda^{(m)}$, $d_\lambda^{(m)}$ 
	given by  
	\begin{align*}
		v_m=\sum_{|\lambda|\leq m}c_\lambda^{(m)}F_{-\lambda}|\Lambda'\rangle,\quad 
		u_m=\sum_{|\lambda|\leq m}d_\lambda^{(m)}F_{-\lambda}|\Lambda'\rangle
		\end{align*}
	are polynomials in $c$, $\Delta$, $\beta_p$, 
	$\Lambda_p,\ldots, \Lambda_{2p}, \Lambda_{2p}^{-1}$. Moreover, $c_\emptyset^{(2k+1)}=d_\emptyset^{(2k+1)}=0$ for $k\in\Z_{\geq0}$.  
\end{thm}
\begin{proof}
We use the convention
\[
 v_j=u_j=0\qquad (j<0)
\]
throughout the proof.  We also write
\[
 \mathcal R
 =\mathbb C\bigl[
 c,\Delta,\beta_p,
 \Lambda_p,\ldots,\Lambda_{2p},\Lambda_{2p}^{-1}
 \bigr].
\]

\paragraph{Step 1. Recursive relations.}
We first compare the lowest-order terms in
\eqref{eq:phi_irr} and \eqref{eq:psi_irr}.  For $n\geq0$, the coefficient
of the leading term in the $L_{n+p}$-relation gives
\[
 \Lambda'_{n+p}=\Lambda_{n+p}\qquad (n\geq1),
 \qquad
 \Lambda'_p=\Lambda_p-p\beta_p.
\]
Thus
\[
 \Lambda'_p=\Lambda_p-p\beta_p,
 \qquad
 \Lambda'_n=\Lambda_n\qquad (p+1\leq n\leq2p).
\]
In particular, $\Lambda'_{2p}=\Lambda_{2p}\neq0$.

Set
\[
 \widetilde L_n=L_n-\Lambda'_n,
\]
where $\Lambda'_n=0$ for $n>2p$.  Comparing the coefficients of
$z^{m/2}$ in \eqref{eq:L_phi}--\eqref{eq:G_psi}, after factoring out the
common irregular exponential, shows that
\eqref{eq:L_phi}--\eqref{eq:psi_irr} are equivalent to the following
relations.  For $n\in\mathbb Z_{\geq0}$,
\begin{align}
 \widetilde L_{n+p}v_m
 ={}&\delta_{n,0}p\beta_pv_m
 -\sum_{i=1}^{p}i\beta_i
  v_{m-2(n+p-i)}
 \notag\\
 &+\left(
   \alpha+(n+p+1)\Delta+\frac m2-n-p
  \right)v_{m-2n-2p},
 \label{eq_def_rel_1}\\
 \widetilde L_{n+p}u_m
 ={}&\delta_{n,0}p\beta_pu_m
 -\sum_{i=1}^{p}i\beta_i
  u_{m-2(n+p-i)}
 \notag\\
 &+\left(
   \alpha-\frac{p+1}{2}
   +(n+p+1)\left(\Delta+\frac12\right)
   +\frac m2-n-p
  \right)u_{m-2n-2p}.
 \label{eq_def_rel_2}
\end{align}
For $s\in\mathbb Z_{\geq0}+\frac12$,
\begin{align}
 G_{s+p}v_m
 &=u_{m-2s-p},
 \label{eq_def_rel_3}\\
 G_{s+p}u_m
 &=-\sum_{i=1}^{p}i\beta_i
   v_{m-2(s-i)-3p}
 \notag\\
 &\quad
 +\left(
   \alpha-\frac{3p}{2}
   +2\left(s+p+\frac12\right)\Delta
   +\frac m2-s
  \right)v_{m-2s-3p}.
 \label{eq_def_rel_4}
\end{align}
Notice that, in \eqref{eq_def_rel_1} and \eqref{eq_def_rel_2} with
$n=0$, the term $p\beta_pv_m$ or $p\beta_pu_m$ cancels the $i=p$ term
in the corresponding sum. Hence every right-hand side involves only
coefficients of lower index.

\paragraph{Step 2. Determination of the nonconstant coefficients.}
 For $m\geq1$, put
\[
 \mathcal R_{<m}
 =\mathcal R\bigl[
 c_\emptyset^{(j)},d_\emptyset^{(j)}
 \mathrel{}\big|\mathrel{}1\leq j<m
 \bigr].
\]
Steps~2 and~3 are carried out simultaneously by induction on the level $m$.
Before treating level $m\leq2p+1$, we assume that all relations hold at
levels smaller than $m$, that
\[
 \beta_{p-j}\in\mathcal R
 \qquad
 \left(
  1\leq j\leq
  \min\left\{p-1,\left\lfloor\frac{m-1}{2}\right\rfloor\right\}
 \right),
\]
and that every nonconstant coefficient at a lower level has already been
expressed as a polynomial in the scalar coefficients occurring at still
lower levels, with coefficients in $\mathcal R$.  The new parameter
$\beta_{p-m/2}$, when $m$ is even and $m<2p$, and the parameter $\alpha$,
when $m=2p$, will be determined in Step~3.

Putting $m=0$ in \eqref{eq_def_rel_1} and \eqref{eq_def_rel_2}, we obtain
\[
 \Lambda'_p=\Lambda_p-p\beta_p,
 \qquad
 \Lambda'_n=\Lambda_n
 \quad(p+1\leq n\leq2p).
\]
In particular,
\[
 \Lambda'_{2p}=\Lambda_{2p}\neq0.
\]
All Gram matrices below are taken in the target module
$M_{\mathrm{NS}}^{\Lambda',[p]}$.

We first note that any solution of
\eqref{eq_def_rel_1}--\eqref{eq_def_rel_4} satisfies
\[
 v_m,u_m\in U_{\leq m/2}.
\]
Indeed, repeated use of the recursive relations gives
\[
 \widetilde F_\lambda v_m
 =\widetilde F_\lambda u_m=0
 \qquad(|\lambda|>m).
\]
Lemma~\ref{lem_pairing_LL} then implies the asserted degree bounds.  We may
therefore write
\[
 v_m=\sum_{|\mu|\leq m}
 c_\mu^{(m)}F_{-\mu}|\Lambda'\rangle,
 \qquad
 u_m=\sum_{|\mu|\leq m}
 d_\mu^{(m)}F_{-\mu}|\Lambda'\rangle.
\]

For each nonempty NS partition $\lambda$ with $|\lambda|\leq m$, the
relations \eqref{eq_def_rel_1}--\eqref{eq_def_rel_4}, applied successively
in the order defining $\widetilde F_\lambda$, prescribe the constant terms
\[
 b_{\lambda,m}^{(v)}
 =\bigl\{\widetilde F_\lambda v_m\bigr\},
 \qquad
 b_{\lambda,m}^{(u)}
 =\bigl\{\widetilde F_\lambda u_m\bigr\}.
\]
Every vector on the right-hand side has level strictly smaller than $m$.
Moreover, for $m\leq2p+1$, the terms containing $\alpha$ in the positive-mode
relations vanish.  By the induction hypothesis, it follows that
\[
 b_{\lambda,m}^{(v)},b_{\lambda,m}^{(u)}\in\mathcal R_{<m}.
\]
Since $\widetilde F_\lambda|\Lambda'\rangle=0$ for
$\lambda\neq\emptyset$, we have
\begin{align}
 \left(b_{\lambda,m}^{(v)}\right)_{1\leq|\lambda|\leq m}
 &={}
 \left(
  \left\{
   \widetilde F_\lambda F_{-\mu}|\Lambda'\rangle
  \right\}
 \right)_{1\leq|\lambda|,|\mu|\leq m}
 \left(c_\mu^{(m)}\right)_{1\leq|\mu|\leq m},
 \label{eq:full-Gram-v}
 \\
 \left(b_{\lambda,m}^{(u)}\right)_{1\leq|\lambda|\leq m}
 &={}
 \left(
  \left\{
   \widetilde F_\lambda F_{-\mu}|\Lambda'\rangle
  \right\}
 \right)_{1\leq|\lambda|,|\mu|\leq m}
 \left(d_\mu^{(m)}\right)_{1\leq|\mu|\leq m}.
 \label{eq:full-Gram-u}
\end{align}
Corollary~\ref{cor_det} shows that the common coefficient matrix is
invertible.  Hence all nonconstant coefficients are uniquely determined,
and
\begin{equation*}
 c_\mu^{(m)},d_\mu^{(m)}\in\mathcal R_{<m}
 \qquad(1\leq|\mu|\leq m).
\end{equation*}
For $m=1$, all prescribed constant terms vanish, so
\[
 v_1=c_\emptyset^{(1)}|\Lambda'\rangle,
 \qquad
 u_1=d_\emptyset^{(1)}|\Lambda'\rangle.
\]

We next verify that the vectors obtained from
\eqref{eq:full-Gram-v} and \eqref{eq:full-Gram-u} satisfy all positive-mode
relations.  For a basic positive mode $\widetilde F_a$ ($a\geq1$), let
$R_{a,m}^{(v)}$ and $R_{a,m}^{(u)}$ denote the differences between the two
sides of the corresponding relations for $v_m$ and $u_m$, respectively.
The row $\lambda=(a)$ in the full Gram system gives
\[
 \{R_{a,m}^{(v)}\}=\{R_{a,m}^{(u)}\}=0.
\]
More generally, for every NS partition $\lambda$, commuting
$\widetilde F_\lambda$ through the defining relation for
$R_{a,m}^{(v)}$ or $R_{a,m}^{(u)}$ and then using the Neveu--Schwarz
relations expresses
\[
 \bigl\{\widetilde F_\lambda R_{a,m}^{(v)}\bigr\},
 \qquad
 \bigl\{\widetilde F_\lambda R_{a,m}^{(u)}\bigr\}
\]
in terms of relations at lower levels.  These terms vanish by the induction
hypothesis.  The nondegeneracy of the Gram matrix therefore yields
\[
 R_{a,m}^{(v)}=R_{a,m}^{(u)}=0.
\]
Thus \eqref{eq_def_rel_1} and \eqref{eq_def_rel_2} hold for $n\geq1$, and
\eqref{eq_def_rel_3} and \eqref{eq_def_rel_4} hold for
$s\geq\frac12$.

It remains to consider the $n=0$ relations.  Denote the right-hand sides of
\eqref{eq_def_rel_1} and \eqref{eq_def_rel_2} by
$X_{0,m}^{(1)}$ and $X_{0,m}^{(2)}$, respectively, and put
\[
 R_m^{(v)}=\widetilde L_pv_m-X_{0,m}^{(1)},
 \qquad
 R_m^{(u)}=\widetilde L_pu_m-X_{0,m}^{(2)}.
\]
The Neveu--Schwarz relations and the positive-mode relations already proved
give
\begin{align*}
 \widetilde L_{n+p}R_m^{(v)}
 &=\widetilde L_{n+p}R_m^{(u)}=0
 &&(n\geq1),
 \\
 G_{s+p}R_m^{(v)}
 &=G_{s+p}R_m^{(u)}=0
 &&\left(s\geq\frac12\right).
\end{align*}
Corollary~\ref{cor_n_kill} therefore implies
\[
 R_m^{(v)},R_m^{(u)}\in U_{\leq 0}.
\]
Hence the $n=0$ relations hold except possibly for their constant terms.

We record the parity decomposition that will be used to compare the odd and
even levels.  For an NS partition $\lambda$, set
\[
 \epsilon(\lambda)
 =\#\{j\mid \lambda_j\text{ is odd}\}\pmod2.
\]
The operator $\widetilde F_n$ is even when $n$ is even and odd when $n$ is
odd.  The same statement holds for $F_{-n}$.  Consequently,
\[
 \operatorname{par}(\widetilde F_\lambda)
 =\operatorname{par}(F_{-\lambda})
 =\epsilon(\lambda)
 \equiv|\lambda|\pmod2.
\]
We regard $|\Lambda'\rangle$ as even.  It follows that
\begin{equation*}
 \left\{
  \widetilde F_\lambda F_{-\mu}|\Lambda'\rangle
 \right\}=0
 \qquad
 \text{if }\epsilon(\lambda)\neq\epsilon(\mu).
\end{equation*}
Thus the Gram matrix is block diagonal with respect to fermion parity, and
each parity block is invertible.  Let $\pi_{\bar0}$ denote the projection
onto the span of PBW vectors indexed by even-parity NS partitions.
Since $\widetilde L_p$ is even, only $\pi_{\bar0}v_m$ and
$\pi_{\bar0}u_m$ can contribute to the constant terms of the remaining
$L_p$-relations.

\paragraph{Step 3. Determination of
\(\beta_{p-1},\ldots,\beta_1\) and \(\alpha\).}
To treat the two sequences simultaneously, set
\[
 w_m^{(1)}=v_m,
 \qquad
 w_m^{(2)}=u_m,
 \qquad
 e_m^{(1)}=c_\emptyset^{(m)},
 \qquad
 e_m^{(2)}=d_\emptyset^{(m)},
\]
and put $e_0^{(1)}=e_0^{(2)}=1$.  For $k\geq0$, define
$Y_{2k}^{(i)}$ recursively by
\begin{equation}
 Y_{2k}^{(i)}
 =w_{2k}^{(i)}
 -\sum_{j=1}^{k}e_{2j}^{(i)}Y_{2k-2j}^{(i)},
 \qquad
 Y_0^{(i)}=|\Lambda'\rangle,
 \label{eq:def-Y-even}
\end{equation}
and set $Y_{2k}^{(i)}=0$ for $k<0$.  Equivalently,
\begin{equation}
 w_{2k}^{(i)}
 =\sum_{j=0}^{k}e_{2j}^{(i)}Y_{2k-2j}^{(i)}.
 \label{eq:inverse-Y-even}
\end{equation}
Write
\[
 Y_{2k}^{(i)}
 =\sum_\nu y_{2k,\nu}^{(i)}F_{-\nu}|\Lambda'\rangle.
\]
By construction,
\[
 y_{2k,\emptyset}^{(i)}=0
 \qquad(k\geq1).
\]

For $n\geq0$ and $n\leq\ell\leq n+p-1$, put
\[
 q_{n,\ell}=-(n+p-\ell)\beta_{n+p-\ell}.
\]
We also set
\begin{align*}
 A_{n,m}^{(1)}
 &=\alpha+(n+p+1)\Delta+\frac m2-n-p,
 \\
 A_{n,m}^{(2)}
 &=\alpha-\frac{p+1}{2}
 +(n+p+1)\left(\Delta+\frac12\right)
 +\frac m2-n-p.
\end{align*}
Then the two $L$-relations take the uniform form
\begin{equation}
 \widetilde L_{n+p}w_m^{(i)}
 =-\delta_{n,0}q_{0,0}w_m^{(i)}
 +\sum_{\ell=n}^{n+p-1}q_{n,\ell}w_{m-2\ell}^{(i)}
 +A_{n,m}^{(i)}w_{m-2n-2p}^{(i)}.
 \label{eq:uniform-L-recursion}
\end{equation}
For $n=0$, the first term cancels the summand with $\ell=0$.  Moreover,
\[
 A_{0,m}^{(1)}=A_{0,m}^{(2)}
 =\alpha+(p+1)\Delta+\frac m2-p.
\]
We denote this common value by $A_{0,m}$.

\begin{claim}
For $i=1,2$, the following identities hold.
\begin{enumerate}
\item If $1\leq k\leq p$ and $n\geq1$, then
\begin{equation}
 \widetilde L_{n+p}Y_{2k}^{(i)}
 =\sum_{\ell=n}^{n+p-1}
 q_{n,\ell}Y_{2k-2\ell}^{(i)}.
 \label{eq_LY1}
\end{equation}

\item If $1\leq k\leq p-1$, then
\begin{equation}
 \widetilde L_pY_{2k}^{(i)}
 =\sum_{\ell=1}^{k}
 q_{0,\ell}Y_{2k-2\ell}^{(i)}.
 \label{eq_LY2}
\end{equation}

\item At level $2p$,
\begin{equation}
 \widetilde L_pY_{2p}^{(i)}
 =\sum_{\ell=1}^{p-1}
 q_{0,\ell}Y_{2p-2\ell}^{(i)}
 +A_{0,2p}Y_0^{(i)}.
 \label{eq_LY3}
\end{equation}
\end{enumerate}
\end{claim}

\begin{proof}[Proof of the claim]
We argue by induction on $k$.  The proof is the same for $i=1$ and $i=2$.
For part~(1), the last term in \eqref{eq:uniform-L-recursion} vanishes,
since
\[
 2k-2n-2p<0
 \qquad(k\leq p,\ n\geq1).
\]
Using \eqref{eq:def-Y-even}, the induction hypothesis, and then
\eqref{eq:inverse-Y-even}, we obtain
\begin{align*}
 \widetilde L_{n+p}Y_{2k}^{(i)}
 &={}
 \sum_{\ell=n}^{n+p-1}q_{n,\ell}w_{2k-2\ell}^{(i)}
 -\sum_{j=1}^{k}e_{2j}^{(i)}
  \sum_{\ell=n}^{n+p-1}
  q_{n,\ell}Y_{2k-2j-2\ell}^{(i)}
 \\
 &=\sum_{\ell=n}^{n+p-1}
 q_{n,\ell}Y_{2k-2\ell}^{(i)}.
\end{align*}
For part~(2), the last term in \eqref{eq:uniform-L-recursion} again
vanishes because $2k-2p<0$.  The same calculation, now with $n=0$ after
the cancellation of the $\ell=0$ term, gives \eqref{eq_LY2}.  At
$k=p$, the calculation leaves the additional term
$A_{0,2p}Y_0^{(i)}$, which gives \eqref{eq_LY3}.
\end{proof}

We now make the parity comparison between consecutive even and odd levels
explicit.

\begin{claim}
For $0\leq k\leq p$, the following statements hold.
\begin{enumerate}
\item The even-parity components satisfy
\begin{equation}
 \pi_{\bar0}Y_{2k}^{(1)}
 =\pi_{\bar0}Y_{2k}^{(2)},
 \label{eq:even-Y-equality}
\end{equation}
and every nonconstant PBW coefficient of this vector belongs to
$\mathcal R$.

\item For $i=1,2$,
\begin{equation}
 \pi_{\bar0}w_{2k+1}^{(i)}
 =\sum_{j=0}^{k}e_{2j+1}^{(i)}
  \pi_{\bar0}Y_{2k-2j}^{(i)}.
 \label{eq:odd-even-shift}
\end{equation}
Thus the even-parity coefficient system at level $2k+1$ is obtained from
that at level $2k$ by replacing
\[
 (e_0^{(i)},e_2^{(i)},\ldots,e_{2k}^{(i)})
 \quad\text{with}\quad
 (e_1^{(i)},e_3^{(i)},\ldots,e_{2k+1}^{(i)}).
\]
\end{enumerate}
\end{claim}

\begin{proof}[Proof of the claim]
First observe that each use of a \(G\)-relation lowers the coefficient
index by at least \(p+1\).  Indeed,
\eqref{eq_def_rel_3} replaces \(v_m\) by \(u_{m-2s-p}\), whereas every
index on the right-hand side of \eqref{eq_def_rel_4} is at most
\(m-2s-p\); here \(2s\geq1\).  Since the \(L\)-relations also lower the
coefficient index, any recursive evaluation involving at least two
\(G\)-relations vanishes when \(m\leq2p+1\).

Let \(\lambda\) be an even-parity NS partition.  If \(\lambda\) contains
an odd part, then it contains at least two odd parts, and hence
\[
 \bigl\{\widetilde F_\lambda Y_{2k}^{(i)}\bigr\}=0
 \qquad(i=1,2).
\]
Thus only partitions consisting entirely of even parts can give nonzero
right-hand sides in the even-parity Gram system.  For such partitions, the recursive calculation uses only the
\(L\)-relations.  Hence \eqref{eq:even-Y-equality} follows from
\eqref{eq_LY1} and the invertibility of the even-parity block of the Gram
matrix. The same argument also shows that all nonconstant coefficients of
\(\pi_{\bar0}Y_{2k}^{(i)}\) belong to \(\mathcal R\).

For part~(2), set
\[
 Z_{2k+1}^{(i)}
 =\pi_{\bar0}w_{2k+1}^{(i)}
 -\sum_{j=0}^{k}e_{2j+1}^{(i)}
  \pi_{\bar0}Y_{2k-2j}^{(i)}.
\]
Its constant term is zero.  The same recursive comparison shows that
\[
 \bigl\{\widetilde F_\lambda Z_{2k+1}^{(i)}\bigr\}=0
\]
for every nonempty even-parity partition $\lambda$: in the range under
consideration, the positive-mode recursions at levels $2k$ and $2k+1$
differ only by the displayed shift of the scalar coefficients.  For an
odd-parity partition $\lambda$, the same pairing vanishes by parity.
The full Gram matrix is nondegenerate, and hence $Z_{2k+1}^{(i)}=0$.  This
proves \eqref{eq:odd-even-shift}.
\end{proof}

We can now determine the remaining parameters.  Let $1\leq k\leq p-1$.
Since $\widetilde L_p$ preserves fermion parity, the constant term of
\eqref{eq_LY2} depends only on $\pi_{\bar0}Y_{2k}^{(i)}$ and therefore
belongs to $\mathcal R$.  Since $Y_{2r}^{(i)}$ has zero constant term for
$r\geq1$, taking the constant term of \eqref{eq_LY2} gives
\[
 \bigl\{\widetilde L_pY_{2k}^{(i)}\bigr\}
 =q_{0,k}=-(p-k)\beta_{p-k}.
\]
The parameter $\beta_{p-k}$ does not occur in the positive-mode equations
used to construct $Y_{2k}^{(i)}$ and appears here for the first time.
Because $p-k\neq0$, this equation uniquely determines
\[
 \beta_{p-k}\in\mathcal R.
\]
The two equations, for $i=1$ and $i=2$, agree by
\eqref{eq:even-Y-equality}.  Proceeding successively for
$k=1,\ldots,p-1$ determines
\[
 \beta_{p-1},\ldots,\beta_1.
\]

At level $2p$, taking the constant term of \eqref{eq_LY3} gives
\[
 \bigl\{\widetilde L_pY_{2p}^{(i)}\bigr\}
 =A_{0,2p}=\alpha+(p+1)\Delta.
\]
The parameter $\alpha$ does not occur in the positive-mode equations used
to construct $Y_{2p}^{(i)}$, and hence this equation uniquely determines
\[
 \alpha\in\mathcal R.
\]
Again, the equations for $i=1,2$ agree.

It remains to verify the constant terms of the $L_p$-relations at odd
levels.  Let $0\leq k\leq p-1$.  Using
\eqref{eq:odd-even-shift}, \eqref{eq_LY2}, and the fact that odd-parity
vectors have zero constant term after applying $\widetilde L_p$, we obtain
\begin{align*}
 \bigl\{\widetilde L_pw_{2k+1}^{(i)}\bigr\}
 &={}
 \sum_{j=0}^{k}e_{2j+1}^{(i)}
 \bigl\{\widetilde L_pY_{2k-2j}^{(i)}\bigr\}
 \\
 &={}
 \sum_{\ell=1}^{p-1}q_{0,\ell}
 \bigl\{w_{2k+1-2\ell}^{(i)}\bigr\}.
\end{align*}
This is the constant term of the $L_p$-relation at level $2k+1$; the term
with coefficient $A_{0,2k+1}$ is absent because $2k+1-2p<0$.  Hence the
odd levels $1,3,\ldots,2p-1$ impose no additional conditions.

At level $2p+1$, equations \eqref{eq:odd-even-shift}, \eqref{eq_LY2}, and
\eqref{eq_LY3} give
\[
 \bigl\{\widetilde L_pw_{2p+1}^{(i)}\bigr\}
 -\sum_{\ell=1}^{p-1}q_{0,\ell}
  \bigl\{w_{2p+1-2\ell}^{(i)}\bigr\}
 =A_{0,2p}e_1^{(i)}.
\]
On the other hand, the required $L_p$-relation has
$A_{0,2p+1}e_1^{(i)}$ on the right-hand side.  Since
\[
 A_{0,2p+1}-A_{0,2p}=\frac12,
\]
we obtain
\[
 e_1^{(1)}=e_1^{(2)}=0,
\]
that is,
\[
 c_\emptyset^{(1)}=d_\emptyset^{(1)}=0.
\]

We have therefore determined
$\beta_1,\ldots,\beta_{p-1}$ and $\alpha$ as elements of $\mathcal R$, and
all relations hold through level $2p+1$.  More explicitly, for
$1\leq m\leq2p+1$ and $\lambda\neq\emptyset$,
\begin{equation*}
 c_\lambda^{(m)},d_\lambda^{(m)}
 \in
 \mathcal R\bigl[
 c_\emptyset^{(2)},\ldots,c_\emptyset^{(m-1)},
 d_\emptyset^{(2)},\ldots,d_\emptyset^{(m-1)}
 \bigr],
\end{equation*}
where an empty list of variables is omitted.  The vectors $v_m$ and $u_m$
themselves also contain their scalar coefficients
$c_\emptyset^{(m)}$ and $d_\emptyset^{(m)}$.  These scalar coefficients,
for $m\geq2$, have not yet been assigned values: they must be retained as
variables in all intermediate polynomial expressions.  The coefficient
$c_\emptyset^{(q)}$, respectively $d_\emptyset^{(q)}$, is determined only
when the constant term of the $L_p$-relation at level $q+2p$ is imposed in
Steps~4 and~5.  After those values are substituted, all coefficients will
belong to $\mathcal R$.

\paragraph{Step 4. Determination of the constant terms for
$m\geq2p+2$.}
Assume that $m\geq2p+2$, that all lower-level vectors have already been constructed as polynomial expressions in the still-undetermined scalar coefficients, and that
\[
 c_\emptyset^{(1)},\ldots,c_\emptyset^{(m-2p-1)},
 \qquad
 d_\emptyset^{(1)},\ldots,d_\emptyset^{(m-2p-1)}
\]
belong to $\mathcal R$.  As in Step~2, the nonconstant coefficients at
level $m$ are uniquely determined so that all relations hold except
possibly for the constant terms of the two $n=0$ equations.  We must show
that these remaining equations determine
$c_\emptyset^{(m-2p)}$ and $d_\emptyset^{(m-2p)}$.

We first treat the even case $m=2k$.  We continue to use
$w_m^{(i)}$, $e_m^{(i)}$, and $Y_{2k}^{(i)}$ from Step~3.

\begin{claim}
For $i=1,2$, the following statements hold.
\begin{enumerate}
\item If $k\geq p+1$ and $n\geq1$, then
\begin{equation}
 \widetilde L_{n+p}Y_{2k}^{(i)}
 =\sum_{\ell=n}^{n+p-1}
 q_{n,\ell}Y_{2k-2\ell}^{(i)}
 +f_{n,k}^{(i)},
 \label{eq_LY_n_m_2p}
\end{equation}
where $f_{n,k}^{(i)}$ is a linear combination of
\[
 w_0^{(i)},w_2^{(i)},\ldots,w_{2k-2n-2p}^{(i)}
\]
with coefficients in $\mathcal R$. 

\item If $k\geq p+1$, then
\begin{equation}
 \widetilde L_pY_{2k}^{(i)}
 =\sum_{\ell=1}^{p-1}
 q_{0,\ell}Y_{2k-2\ell}^{(i)}
 +A_{0,2k}w_{2k-2p}^{(i)}
 -e_{2k-2p}^{(i)}A_{0,2p}w_0^{(i)}
 +g_k^{(i)},
 \label{eq_LY_0_m_2p}
\end{equation}
where $g_k^{(i)}$ is a linear combination of
\[
 w_0^{(i)},w_2^{(i)},\ldots,w_{2k-2p-2}^{(i)}
\]
with coefficients in $\mathcal R$. 
\end{enumerate}
\end{claim}

\begin{proof}[Proof of the claim]
Both statements follow by induction on $k$.  For part~(1), substituting
the definition of $Y_{2k}^{(i)}$ and using the induction hypothesis gives
\begin{align*}
 \widetilde L_{n+p}Y_{2k}^{(i)}
 =\sum_{\ell=n}^{n+p-1}
 q_{n,\ell}Y_{2k-2\ell}^{(i)}
 +A_{n,2k}^{(i)}w_{2k-2n-2p}^{(i)}
 -\sum_{j=1}^{k-n-p}
 e_{2j}^{(i)}f_{n,k-j}^{(i)}.
\end{align*}
The last two terms have precisely the asserted form and define
$f_{n,k}^{(i)}$.

For part~(2), we use \eqref{eq_LY2} and \eqref{eq_LY3} when the lower
index is at most $2p$, and the induction hypothesis otherwise.  After
collecting the terms containing $q_{0,\ell}$, we obtain
\begin{align*}
 \widetilde L_pY_{2k}^{(i)}
 =\sum_{\ell=1}^{p-1}
 q_{0,\ell}Y_{2k-2\ell}^{(i)}
 +A_{0,2k}w_{2k-2p}^{(i)}
 -e_{2k-2p}^{(i)}A_{0,2p}w_0^{(i)}
 +g_k^{(i)},
\end{align*}
where the remaining term $g_k^{(i)}$ involves only lower-level vectors
and previously determined scalar coefficients.  This proves the claim.
\end{proof}

Equation~\eqref{eq_LY_n_m_2p}, together with the same parity argument as
in Step~3, shows that the even-parity coefficients of $Y_{2k}^{(i)}$ are
polynomials in the data already determined.  We may therefore take the
constant term of \eqref{eq_LY_0_m_2p}.  The only new unknown is
$e_{2k-2p}^{(i)}$, and its coefficient is
\[
 A_{0,2k}-A_{0,2p}=k-p\neq0.
\]
Thus $e_{2k-2p}^{(i)}$ is uniquely determined and belongs to
$\mathcal R$.  Equivalently, for every even $m\geq2p+2$,
\[
 c_\emptyset^{(m-2p)},
 \ d_\emptyset^{(m-2p)}\in\mathcal R
\]
are uniquely determined.

\paragraph{Step 5. Vanishing of the odd constant terms.}
We prove by induction on the positive odd integer \(q\) that
\[
 \pi_{\bar0}w_q^{(1)}
 =
 \pi_{\bar0}w_q^{(2)}
 =0.
\]
The case \(q=1\) was proved in Step~3.  Let \(q\geq3\) be odd, and
assume that
\[
 \pi_{\bar0}w_{q'}^{(1)}
 =
 \pi_{\bar0}w_{q'}^{(2)}
 =0
\]
for every positive odd integer \(q'<q\).

Let \(\lambda\) be a nonempty even-parity NS partition.  Successive
application of the recursive relations expresses
\[
 \bigl\{\widetilde F_\lambda w_q^{(i)}\bigr\}
\]
in terms of even-parity components of vectors \(w_{q'}^{(j)}\) with
odd \(q'<q\).  Hence, by the induction hypothesis,
\[
 \bigl\{\widetilde F_\lambda w_q^{(i)}\bigr\}=0
 \qquad(i=1,2).
\]
The even-parity block of the Gram matrix is invertible.  Therefore all
nonconstant even-parity coefficients of \(w_q^{(i)}\) vanish, and
\[
 \pi_{\bar0}w_q^{(i)}
 =
 e_q^{(i)}|\Lambda'\rangle.
\]

We now take the constant term of the \(L_p\)-relation at level
\(q+2p\).  The same calculation as in Step~3, with the scalar indices
shifted by \(q\), shows that the constant term determined by the
nonconstant coefficients is
\[
 \sum_{\ell=1}^{p-1}
 q_{0,\ell}
 \bigl\{w_{q+2p-2\ell}^{(i)}\bigr\}
 +A_{0,2p}e_q^{(i)}.
\]
On the other hand, the required \(L_p\)-relation gives
\[
 \sum_{\ell=1}^{p-1}
 q_{0,\ell}
 \bigl\{w_{q+2p-2\ell}^{(i)}\bigr\}
 +A_{0,q+2p}e_q^{(i)}.
\]
It follows that
\[
 \bigl(A_{0,q+2p}-A_{0,2p}\bigr)e_q^{(i)}=0.
\]
Since
\[
 A_{0,q+2p}-A_{0,2p}=\frac q2\neq0,
\]
we obtain
\[
 e_q^{(i)}=0
 \qquad(i=1,2).
\]
Consequently,
\[
 \pi_{\bar0}w_q^{(1)}
 =
 \pi_{\bar0}w_q^{(2)}
 =0,
\]
which completes the induction.

Steps~2--5 determine all coefficients $v_m$ and $u_m$ uniquely.
They also determine $\beta_1,\ldots,\beta_{p-1}$ and $\alpha$ uniquely.
At each stage, the only divisions are by a positive integral multiple of
$\Lambda_{2p}$ or by one of the nonzero integers occurring above.
Consequently, all coefficients belong to $\mathcal R$.

The resulting series \eqref{eq:phi_irr} and \eqref{eq:psi_irr} satisfy
\eqref{eq_def_rel_1}--\eqref{eq_def_rel_4}, and hence the defining
commutation relations \eqref{eq:L_phi}--\eqref{eq:G_psi}.  Extending the
action from the generating irregular vector by these commutation relations
defines the two operators on the induced module
$M_{\mathrm{NS}}^{\Lambda,[p]}$.  The triangular determination above also
shows that no other choice of coefficients is possible.  This proves both
existence and uniqueness.
\end{proof}

\begin{re}
The irregular vertex operators constructed in this section should be
regarded as rank-zero irregular vertex operators.  Indeed, their
commutation relations with the Neveu--Schwarz generators coincide with
those of ordinary vertex operators, and the source and target irregular
Verma modules have the same rank.  By contrast, in the Virasoro case,
irregular vertex operators between irregular Verma modules of different
ranks have also been constructed
\cite{Nagoya2015,Nagoya2018}.  We expect analogous rank-changing
irregular vertex operators to exist for the Neveu--Schwarz algebra.
\end{re}

\section{Decomposition}
In this section, we prove a decomposition theorem for the tensor product
of the free-fermion Fock space with a Neveu--Schwarz irregular Verma module.
The tensor product decomposes into an infinite direct sum of irregular Verma
modules for \(\mathrm{Vir}\oplus\mathrm{Vir}\), generalizing the regular
decomposition theorem of \cite{BBFLT}.  We also prove the corresponding
decomposition of Neveu--Schwarz irregular vertex operators into tensor
products of Virasoro irregular vertex operators. 

\subsection{Decomposition of the irregular Verma module}

The Virasoro algebra 
\begin{equation*}
	\Vir=\bigoplus_{n\in\Z}\C L_n\oplus \C c
\end{equation*}
is a Lie algebra with commutation relations
\begin{align*}
	&[L_m,L_n]=(m-n)L_{m+n}+\frac{c}{12}(m^3-m)\delta_{m+n,0}, 
\\ &[\Vir, c]=0.
\end{align*}
Set 
\begin{equation*}
	\Vir_p=\bigoplus_{n\geq p}\C L_n\oplus \C c. 
\end{equation*}
For $\Lambda=(\Lambda_{p},\Lambda_{p+1},\ldots,\Lambda_{2p}) \in\C^{p+1}$, 
let $\C|\Lambda\rangle$ be one-dimensional $\Vir_p$-module with 
\begin{equation*}
	L_n|\Lambda\rangle=\Lambda_n|\Lambda\rangle\quad (n\geq p), 
\end{equation*}
where $\Lambda_n=0$ if $n>2p$. The element $c$ acts as 
multiplication by a complex number. 
Define the induced module 
\begin{equation*}
	M^{\Lambda,[p]}_{\Vir}=\mathrm{Ind}^{\Vir}_{\Vir_p}\C |\Lambda\rangle. 
\end{equation*}
We call this module a rank-$p$ irregular Verma module of the Virasoro algebra 
and $|\Lambda\rangle$ an irregular vector. 

Let $\mathrm{F}$ be the free-fermion algebra with the generators $f_r$ ($r\in\Z+1/2$) and 
anticommutation relations
\begin{align*}
&[f_r,f_s]_+=\delta_{r+s,0}. 
\end{align*}
Let $M_\mathrm{F}$ be the fermion Fock space with the vacuum vector 
$|1\rangle$ such that 
\begin{align*}
&f_r|1\rangle=0\quad (r>0).  
\end{align*}
The $\mathrm{F}\oplus \mathrm{NS}$ algebra is a direct sum of the free-fermion algebra 
$\mathrm{F}$ and the Neveu--Schwarz algebra $\mathrm{NS}$ with cross-commutation and cross-anticommutation relations
\begin{equation*}
	[f_r,L_n]=[f_r,G_s]_+=0. 
\end{equation*}
Let us denote by $M_{\mathrm{F}\oplus \mathrm{NS}}^{\Lambda,[p]}$ an irregular Verma 
module of the $\mathrm{F}\oplus \mathrm{NS}$ algebra, which is isomorphic 
to a tensor product of $M_\mathrm{F}$ and $M_\mathrm{NS}^{\Lambda,[p]}$. 
Its irregular vector is $|1\rangle\otimes |\Lambda\rangle$. 
When $p=0$, $M_{\mathrm{F}\oplus \mathrm{NS}}^{\Lambda,[0]}$ 
decomposes into a direct sum of 
the Verma modules of a direct sum of two Virasoro algebras \cite{BBFLT}. 
In this subsection, we generalize it to the irregular case.

We use the free-field realization of the Neveu--Schwarz algebra generated 
by $a_n$ ($n\in\Z$) and  $\psi_r$ ($r\in\Z+\frac{1}{2}$) with relations
\begin{equation*}
	[a_n,a_m]=n\delta_{n+m,0}, \quad [a_n,\psi_r]=0,\quad \{\psi_r,\psi_s\}=\delta_{r+s,0}.
\end{equation*}

For $P=(P_0,\ldots, P_p)\in\C^{p+1}$, let $\mathcal{F}^{[p]}_P$ be an 
irregular Fock space of this algebra generated 
by a vacuum vector $|P\rangle$ such that  
\begin{equation*}
a_n|P\rangle=P_n|P\rangle\quad (n=0,\ldots,p),\quad a_n|P\rangle=\psi_r|P\rangle=0 \quad \left (n>p,\ r > 0\right). 
\end{equation*}

Set 
\begin{align*}
	c=1+2Q^2\quad (Q=b^{-1}+b). 
\end{align*}

\begin{prop}
The following formulas define two representations of the Neveu--Schwarz algebra on the irregular Fock space \(\mathcal F_P^{[p]}\).
\begin{align*}
	&L_n=\frac{1}{2}\sum_{k\neq 0,n}a_ka_{n-k}+\frac{1}{2}\sum_{r}\left(r-\frac{n}{2}\right)\psi_{n-r}\psi_r
	+\frac{i}{2}\left(Qn-2a_0\right)a_n
	\quad (n\neq 0),
	\\
	&L_0=\sum_{k>0}a_{-k}a_k+\sum_{r>0}r\psi_{-r}\psi_r+\frac{1}{2}\left(\frac{Q^2}{4}-a_0^2\right),
	\\
	&G_r=\sum_{k\neq 0}a_k\psi_{r-k}+i(Qr-a_0)\psi_r,
\end{align*}
and
\begin{align*}
&L_n=\frac{1}{2}\sum_{k\neq 0,n}a_ka_{n-k}+\frac{1}{2}\sum_{r}\left(r-\frac{n}{2}\right)\psi_{n-r}\psi_r
-\frac{i}{2}\left(Q(n-2p)+2a_0\right)a_n
\quad (n\neq 0),
\\
&L_0=\sum_{k>0}a_{-k}a_k+\sum_{r>0}r\psi_{-r}\psi_r+\frac{1}{2}\left(\frac{Q^2}{4}-(Qp-a_0)^2\right),
\\
&G_r=-\sum_{k\neq 0}a_k\psi_{r-k}+i(Q(r-p)+a_0)\psi_r.
\end{align*}
Under both representations, the vector \(|P\rangle\) is an irregular vector with the same irregular weight \(\Lambda\):
\begin{equation}\label{eq_Lambda_P}
	\Lambda_n=\frac{1}{2}\sum_{k=n-p}^p P_kP_{n-k}\quad (p<n\leq 2p),\quad 
	\Lambda_{p}=\frac{1}{2}\sum_{k=1}^{p-1} P_kP_{p-k}+\frac{i}{2}(Qp- 2P_0)P_p. 
\end{equation}
\end{prop}
\begin{proof}
The first set of formulas is the standard free-field realization of the
NS algebra. For \(p=0\), it coincides with the realization
used in \cite[(2.2)]{BS}.  Its verification uses only the oscillator
relations
\[
 [a_n,a_m]=n\delta_{n+m,0},
 \qquad
 [a_n,\psi_r]=0,
 \qquad
 [\psi_r,\psi_s]_+=\delta_{r+s,0},
\]
and is therefore independent of the choice of the cyclic vector
\(|P\rangle\).  In particular, the first set of operators satisfies the
Neveu--Schwarz relations with central charge
\[
 c=1+2Q^2.
\]

Define an automorphism \(\sigma_p\) of the oscillator algebra by
\[
 \sigma_p(a_n)=-a_n\quad(n\neq0),
 \qquad
 \sigma_p(a_0)=Qp-a_0,
 \qquad
 \sigma_p(\psi_r)=\psi_r.
\]
Since \(a_0\) is central, the map \(\sigma_p\) preserves all oscillator
relations.  Moreover, \(\sigma_p^2=\mathrm{id}\).

Applying \(\sigma_p\) to the first free-field realization gives
\begin{align*}
 \sigma_p(L_n)
 ={}&
 \frac12\sum_{k\neq0,n}a_ka_{n-k}
 +\frac12\sum_r
 \left(r-\frac n2\right)\psi_{n-r}\psi_r
 \\
 &\quad
 -\frac{i}{2}\left(Q(n-2p)+2a_0\right)a_n
 \qquad(n\neq0),
\end{align*}
together with
\[
 \sigma_p(L_0)
 =
 \sum_{k>0}a_{-k}a_k
 +\sum_{r>0}r\psi_{-r}\psi_r
 +\frac12\left(\frac{Q^2}{4}-(Qp-a_0)^2\right),
\]
and
\[
 \sigma_p(G_r)
 =
 -\sum_{k\neq0}a_k\psi_{r-k}
 +i\bigl(Q(r-p)+a_0\bigr)\psi_r.
\]
Thus \(\sigma_p\) transforms the first free-field realization into the
second.  Since \(\sigma_p\) is an automorphism of the oscillator algebra,
the second realization also satisfies the Neveu--Schwarz relations.

We now compute the action on \(|P\rangle\).  If \(n\geq p\), the
fermionic quadratic term in \(L_n\) annihilates \(|P\rangle\).  A bosonic
term \(a_ka_{n-k}|P\rangle\) can be nonzero only when
\[
 1\leq k\leq p,
 \qquad
 1\leq n-k\leq p.
\]
Consequently, for \(p<n\leq2p\),
\[
 L_n|P\rangle
 =
 \frac12\sum_{k=n-p}^{p}P_kP_{n-k}|P\rangle,
\]
whereas
\[
 L_n|P\rangle=0
 \qquad(n>2p).
\]
For \(n=p\), the linear term also contributes, and we obtain
\[
 L_p|P\rangle
 =
 \left(
  \frac12\sum_{k=1}^{p-1}P_kP_{p-k}
  +\frac{i}{2}(Qp-2P_0)P_p
 \right)|P\rangle.
\]

If \(r>p\), then each term in
\(\sum_{k\neq0}a_k\psi_{r-k}|P\rangle\) vanishes.  Indeed, if
\(k\leq p\), then \(r-k>0\), while if \(k>p\), the mode \(a_k\)
annihilates \(|P\rangle\).  The remaining term is also zero, since
\(\psi_r|P\rangle=0\).  Thus
\[
 G_r|P\rangle=0
 \qquad(r>p).
\]

The same formulas hold in the second realization.  The quadratic parts
are unchanged, and the only additional check is the linear contribution
to \(L_p|P\rangle\), for which
\[
 -\frac{i}{2}
 \bigl(Q(p-2p)+2P_0\bigr)P_p
 =
 \frac{i}{2}(Qp-2P_0)P_p.
\]
Therefore, the vector \(|P\rangle\) has the same irregular weight~\eqref{eq_Lambda_P} in both
realizations.
\end{proof}

The condition $P_p\neq0$ implies $\Lambda_{2p}\neq0$.  Hence the natural
homomorphism
\[
 M_{\mathrm{NS}}^{\Lambda,[p]}
 \longrightarrow
 \mathcal F_P^{[p]}
\]
is injective, with image equal to the NS-submodule generated by
\(|P\rangle\).  Lemma~\ref{lem:free-field-filtration} below shows that this
image is the entire irregular Fock space.  We denote the corresponding
free-field operators by $a_n^{(i)}$ and $\psi_r^{(i)}$, where $i=1,2$. 
For example, 
\begin{equation}\label{eq:psi-G-relation}
	\psi_{-1/2}^{(1)}|P\rangle=\frac{1}{P_p}G_{-1/2+p}|P\rangle,\quad 
	\psi_{-1/2}^{(2)}|P\rangle=-\frac{1}{P_p}G_{-1/2+p}|P\rangle. 
\end{equation}

The embedding of the $\Vir\oplus\Vir$ subalgebra in the $\mathrm{F}\oplus\mathrm{NS}$ algebra is defined by
\begin{align*}
	L_n^{(1)}=&\frac{b^{-1}}{b^{-1}-b}L_n-\frac{b^{-1}+2b}{2(b^{-1}-b)}\sum_{r\in\Z+\frac{1}{2}}r:f_{n-r}f_r:
	+\frac{1}{b^{-1}-b}\sum_{r\in\Z+\frac{1}{2}}f_{n-r}G_r,
	\\
	L_n^{(2)}=&\frac{b}{b-b^{-1}}L_n-\frac{b+2b^{-1}}{2(b-b^{-1})}\sum_{r\in\Z+\frac{1}{2}}r:f_{n-r}f_r:
	+\frac{1}{b-b^{-1}}\sum_{r\in\Z+\frac{1}{2}}f_{n-r}G_r. 
\end{align*} 
In what follows, we assume that 
\begin{equation*}
    b\neq 0, \quad b^2\neq 1. 
\end{equation*}
The central charges of these $\Vir^{(1)}$ and $\Vir^{(2)}$ subalgebras are equal to
\begin{equation*}
	c^{(i)}=1+6\left(b^{(i)}+\frac{1}{b^{(i)}}\right)^2\quad 
   (i=1,2), 
\end{equation*}
where
\begin{equation*}
    b^{(1)}=b\sqrt{\frac{2}{1-b^2}}, \quad 
    b^{(2)}=b^{-1}\sqrt{\frac{2}{1-b^{-2}}}.
\end{equation*}

For $p\leq n\leq 2p$, $L_n^{(i)}$ ($i=1,2$) act on $|P\rangle$ as follows:  
\begin{align*}
&L_n^{(1)}|P\rangle=\frac{b^{-1}}{b^{-1}-b}\Lambda_n|P\rangle,
\quad 
L_n^{(2)}|P\rangle=\frac{b}{b-b^{-1}}\Lambda_n|P\rangle.  
\end{align*}
Hence if $P_p\neq 0$, the cyclic submodule generated by $|P\rangle$ is isomorphic to the tensor product  
$M^{\Lambda^{(1)},[p]}_{\Vir}\otimes M^{\Lambda^{(2)},[p]}_{\Vir}$ 
of rank-$p$ irregular Verma modules with weights  
\begin{equation*}
	\Lambda^{(1)}=\frac{b^{-1}}{b^{-1}-b}\Lambda,\quad 
\Lambda^{(2)}=\frac{b}{b-b^{-1}}\Lambda. 
\end{equation*}

Put 
\begin{align*}
	\chi_r^{(j)}=f_r-i\psi_r^{(j)}\quad (j=1,2).
\end{align*}

We define the vectors $|P,n\rangle$ ($2n\in\Z$) as follows. 
$|P,0\rangle=|1\rangle\otimes |P\rangle=|1\rangle\otimes |\Lambda\rangle$, 
\begin{align*}
    &|P,n\rangle
=
\chi^{(1)}_{-(2n-1/2)}
\cdots
\chi^{(1)}_{-3/2}
\chi^{(1)}_{-1/2}|P,0\rangle
\qquad(n>0),
\\
&|P,n\rangle
=
\chi^{(2)}_{-(2|n|-1/2)}
\cdots
\chi^{(2)}_{-3/2}
\chi^{(2)}_{-1/2}|P,0\rangle
\qquad(n<0). 
\end{align*}

\begin{exmp}\label{exmp_Pn}
	When $p=1$, the vectors $|P,1\rangle$ and $|P,-1\rangle$ are realized in 
$M_{\mathrm{F}\oplus \mathrm{NS}}^{\Lambda,[1]}$   
as 
	\begin{align*}
		|P,1\rangle=&\left(f_{-3/2}f_{-1/2}-\frac{i}{P_1}f_{-3/2}G_{1/2}
		-\frac{i}{P_1}G_{-1/2}f_{-1/2}-\frac{1}{P_1^2}G_{-1/2}
		G_{1/2}\right.
		\\
		&\left.+\frac{1}{P_1^2}L_0-\frac{Q+2P_0}{2P_1^2}f_{-1/2}
		G_{1/2}-\frac{(Q+2P_0)(Q-2P_0)}{8P_1^2}\right)|P,0\rangle, 
		\\
		|P,-1\rangle=&\left(f_{-3/2}f_{-1/2}+\frac{i}{P_1}f_{-3/2}G_{1/2}
		+\frac{i}{P_1}G_{-1/2}f_{-1/2}-\frac{1}{P_1^2}G_{-1/2}
		G_{1/2}\right. 
		\\
		&\left.+\frac{1}{P_1^2}L_0-\frac{3Q-2P_0}{2P_1^2}f_{-1/2}
		G_{1/2}+\frac{(3Q-2P_0)(Q-2P_0)}{8P_1^2}\right)|P,0\rangle. 
	\end{align*}
\end{exmp}
\begin{prop}
For $n\geq p$ and $2m\in\mathbb Z$, we have
\begin{align*}
 L_n^{(1)}|P,m\rangle
 &=\left(
   \frac{b^{-1}}{b^{-1}-b}\Lambda_n
   -\delta_{n,p}\frac{2im}{b^{-1}-b}P_p
  \right)|P,m\rangle,
 \\
 L_n^{(2)}|P,m\rangle
 &=\left(
   \frac{b}{b-b^{-1}}\Lambda_n
   -\delta_{n,p}\frac{2im}{b-b^{-1}}P_p
  \right)|P,m\rangle.
\end{align*}
\end{prop}

\begin{proof}
We prove the formulas for \(m\in\frac12\mathbb Z_{\geq0}\) by induction
on \(2m\).

Using the free-field realization, for
\(n\geq p\) and \(s\in\mathbb Z_{<0}+\frac12\), we obtain
\begin{equation}
 \left[L_n^{(1)},\chi_s^{(1)}\right]
 =
 \frac{1}{b^{-1}-b}
 \left[
 \left(
  bs+\frac{n}{2}(b^{-1}+2b)-a_0^{(1)}
 \right)\chi_{n+s}^{(1)}
 -i\sum_{\ell\neq0}
 a_\ell^{(1)}\chi_{n+s-\ell}^{(1)}
 \right].
 \label{eq:L1-chi-commutator}
\end{equation}

For \(m\in\frac12\mathbb Z_{\geq0}\), the definition of
\(|P,m\rangle\) implies
\begin{equation}
 \chi_r^{(1)}|P,m\rangle=0
 \qquad
 \left(r\geq-2m+\frac12\right),
 \label{eq:chi-annihilates-Pm}
\end{equation}
and
\begin{equation}
 a_\ell^{(1)}|P,m\rangle
 =
 P_\ell|P,m\rangle
 \qquad(\ell>0),
 \label{eq:a-positive-Pm}
\end{equation}
where \(P_\ell=0\) for \(\ell>p\).

The first formula for \(m=0\) follows from
\[
 L_n^{(1)}|P,0\rangle
 =
 \frac{b^{-1}}{b^{-1}-b}\Lambda_n|P,0\rangle,
\]
where \(\Lambda_n=0\) for \(n>2p\).

Suppose that the first formula holds for some
\(m\in\frac12\mathbb Z_{\geq0}\), and set
\[
 s=-2m-\frac12.
\]
Then
\[
 |P,m+1/2\rangle
 =
 \chi_s^{(1)}|P,m\rangle.
\]
Let \(n\geq p\).  Since
\[
 n+s\geq-2m+\frac12,
\]
equation~\eqref{eq:chi-annihilates-Pm} gives
\[
 \chi_{n+s}^{(1)}|P,m\rangle=0.
\]
In the sum in~\eqref{eq:L1-chi-commutator}, the term with index
\(\ell<n\) also vanishes.  For \(\ell\geq n\), we may use
\eqref{eq:a-positive-Pm}.  Since \(n\geq p\) and \(P_\ell=0\) for
\(\ell>p\), the only possible nonzero contribution is \(\ell=n=p\).
Therefore,
\begin{equation}
 \left[L_n^{(1)},\chi_s^{(1)}\right]|P,m\rangle
 =
 -\delta_{n,p}\frac{iP_p}{b^{-1}-b}
 \chi_s^{(1)}|P,m\rangle.
 \label{eq:L1-chi-on-Pm}
\end{equation}

Using the induction hypothesis and~\eqref{eq:L1-chi-on-Pm}, we find
\begin{align*}
 L_n^{(1)}|P,m+1/2\rangle
 &=
 \chi_s^{(1)}L_n^{(1)}|P,m\rangle
 +[L_n^{(1)},\chi_s^{(1)}]|P,m\rangle
 \\
 &=
 \left(
  \frac{b^{-1}}{b^{-1}-b}\Lambda_n
  -\delta_{n,p}
   \frac{2i(m+\frac12)}{b^{-1}-b}P_p
 \right)|P,m+1/2\rangle.
\end{align*}
This proves the first formula for \(m\geq0\).  Replacing \(b\) with
\(b^{-1}\) gives the second formula.  Finally, the symmetry
\[
 P_0\longmapsto Qp-P_0,
 \qquad
 P_j\longmapsto-P_j\quad(j\neq0)
\]
gives both formulas for \(m<0\).
\end{proof}

Put
\begin{equation}\label{eq_weights_Vir_decomp}
	\Lambda_n^{(m,1)}=\frac{b^{-1}}{b^{-1}-b}\Lambda_n-\delta_{n,p}\frac{2im}{b^{-1}-b}P_p,
	\quad \Lambda_n^{(m,2)}=\frac{b}{b-b^{-1}}\Lambda_n-\delta_{n,p}\frac{2im}{b-b^{-1}}P_p 
\end{equation}
for $p\leq n\leq 2p$. 

We write
\[
 |P,m\rangle
 =
 |\Lambda^{(m,1)}\rangle\otimes|\Lambda^{(m,2)}\rangle.
\]
The relations above show that \(|P,m\rangle\) is an irregular vector
for \(\Vir\oplus\Vir\) with weights
\(\Lambda^{(m,1)}\) and \(\Lambda^{(m,2)}\).
If \(P_p\neq0\), the corresponding Virasoro irregular Verma modules
are irreducible.  Hence the cyclic submodule generated by
\(|P,m\rangle\) is isomorphic to an irregular Verma module
\(M_{\mathrm{Vir}\oplus\mathrm{Vir}}^{\Lambda,m,[p]}=M_\Vir^{\Lambda^{(m,1)},[p]}\otimes M_\Vir^{\Lambda^{(m,2)},[p]}\).

It remains to show that the sum of these submodules is direct and
exhausts
\(M^{\Lambda,[p]}_{\mathrm{F}\oplus\mathrm{NS}}\).
For this purpose, we introduce a degree filtration and compare the
resulting associated graded module with its counterpart in the
regular case \(p=0\).

For a strictly decreasing sequence
\[
 \mu=(\mu_1>\cdots>\mu_s),
 \qquad
 \mu_i\in\mathbb Z_{\geq 0}+\frac12,
\]
put
\[
 f_{-\mu}=f_{-\mu_1}\cdots f_{-\mu_s},
 \qquad
 |\mu|=\mu_1+\cdots+\mu_s.
\]
We extend the degree filtration of
\(M_{\mathrm{NS}}^{\Lambda,[p]}\) to
\(M_{\mathrm{F}\oplus\mathrm{NS}}^{\Lambda,[p]}\) by setting
\[
 \deg\bigl(f_{-\mu}F_{-\lambda}|P,0\rangle\bigr)
 =
 |\mu|+\frac{|\lambda|}{2}.
\]
For \(d\in\frac12\mathbb Z_{\geq0}\), let
\[
 \mathcal U_{\leq d}^{[p]}
 =
 \left\{
 v\in M_{\mathrm{F}\oplus\mathrm{NS}}^{\Lambda,[p]}
 \mathrel{}\middle|\mathrel{}
 \deg v\leq d
 \right\}.
\]

\begin{lem}\label{lem:free-field-filtration}
The free-field modes are compatible with the filtration
\(\mathcal U_{\leq d}^{[p]}\). More precisely, for \(j=1,2\),
\(k\in\mathbb Z_{>0}\), and
\(r\in\mathbb Z_{>0}+\frac12\), one has
\begin{align*}
 a^{(j)}_{-k}\mathcal U_{\leq d}^{[p]}
 &\subset \mathcal U_{\leq d+k}^{[p]},&
 \bigl(a^{(j)}_k-P_k\bigr)\mathcal U_{\leq d}^{[p]}
 &\subset \mathcal U_{\leq d-k}^{[p]},\\
 \psi^{(j)}_{-r}\mathcal U_{\leq d}^{[p]}
 &\subset \mathcal U_{\leq d+r}^{[p]},&
 \psi^{(j)}_r \mathcal U_{\leq d}^{[p]}
 &\subset \mathcal U_{\leq d-r}^{[p]}.
\end{align*}
Here \(P_k=0\) for \(k>p\), and
\(\mathcal U_{\leq e}^{[p]}=0\) for \(e<0\).
\end{lem}

\begin{proof}
Fix \(j\in\{1,2\}\).  We introduce the free-field filtration
\[
 \mathcal F_{\leq d}^{(j)}
 =
 \operatorname{span}
 \left\{
 f_{-\mu}a^{(j)}_{-\nu}\psi^{(j)}_{-\rho}|P,0\rangle
 \ \middle|\
 |\mu|+|\nu|+|\rho|\leq d
 \right\},
\]
where \(\mu\) and \(\rho\) are strictly decreasing sequences of
positive half-integers and \(\nu\) is a partition into positive
integers.  The oscillator relations immediately give
\begin{align}
 a^{(j)}_{-k}\mathcal F_{\leq d}^{(j)}
 &\subset \mathcal F_{\leq d+k}^{(j)},
 &
 \bigl(a^{(j)}_k-P_k\bigr)\mathcal F_{\leq d}^{(j)}
 &\subset \mathcal F_{\leq d-k}^{(j)},
 \label{eq:c-auxiliary-filtration}\\
 \psi^{(j)}_{-r}\mathcal F_{\leq d}^{(j)}
 &\subset \mathcal F_{\leq d+r}^{(j)},
 &
 \psi^{(j)}_r\mathcal F_{\leq d}^{(j)}
 &\subset \mathcal F_{\leq d-r}^{(j)}.
 \label{eq:psi-auxiliary-filtration}
\end{align}

It remains to compare this filtration with the PBW filtration.
Set
\[
 \varepsilon_1=1,\qquad \varepsilon_2=-1.
\]
The free-field formulas imply that, for every
\(v\in\mathcal F_{\leq d}^{(j)}\) and \(m\in\mathbb Z_{>0}\),
\begin{equation}
 F_{-m}v
 \equiv
 \begin{cases}
  P_p a^{(j)}_{-m/2}v,
       & m\ \text{even},\\[2mm]
  \varepsilon_jP_p\psi^{(j)}_{-m/2}v,
       & m\ \text{odd},
 \end{cases}
 \pmod{\mathcal F_{<d+m/2}^{(j)}}.
 \label{eq:F-free-field-leading-term}
\end{equation}
Indeed, in the even case the only term of filtration degree
\(d+m/2\) comes from the pair of bosonic oscillators
\(a^{(j)}_p\) and \(a^{(j)}_{-m/2}\), whereas in the odd case it
comes from \(a^{(j)}_p\psi^{(j)}_{-m/2}\).  After
\(a^{(j)}_p\) is moved to the right, its scalar part gives \(P_p\);
all commutator terms and all remaining terms have strictly smaller
filtration degree.

Let \(\lambda\) be an NS partition, and denote by
\(\lambda_{\mathrm{ev}}/2\) the partition obtained by dividing its
even parts by \(2\), and by \(\lambda_{\mathrm{odd}}/2\) the strictly
decreasing sequence obtained by dividing its odd parts by \(2\).
Let \(o(\lambda)\) be the number of odd parts of \(\lambda\).
Iterating \eqref{eq:F-free-field-leading-term}, we obtain
\begin{align}
 f_{-\mu}F_{-\lambda}|P,0\rangle
 ={}&
 \varepsilon_j^{o(\lambda)}
 P_p^{\ell(\lambda)}
 f_{-\mu}
 a^{(j)}_{-\lambda_{\mathrm{ev}}/2}
 \psi^{(j)}_{-\lambda_{\mathrm{odd}}/2}|P,0\rangle
 \notag\\
 &\quad
 +\text{terms in }
 \mathcal F_{<|\mu|+|\lambda|/2}^{(j)}.
 \label{eq:PBW-free-field-triangularity}
\end{align}

The PBW basis and the free-field oscillator basis are indexed by the
same data, and the diagonal coefficient in
\eqref{eq:PBW-free-field-triangularity} is
\[
 \varepsilon_j^{o(\lambda)}P_p^{\ell(\lambda)}\neq0.
\]
Hence the change of basis is triangular and invertible in every
degree.  Therefore,
\[
 \mathcal U_{\leq d}^{[p]}=\mathcal F_{\leq d}^{(j)}.
\]
The desired inclusions now follow from
\eqref{eq:c-auxiliary-filtration} and
\eqref{eq:psi-auxiliary-filtration}.
\end{proof}

\begin{lem}\label{lem_graded_decomposition}
Assume that \(P_p\neq0\).

\begin{enumerate}
\item
The associated graded vector space
\(\operatorname{gr}M_{\mathrm{F}\oplus\mathrm{NS}}^{\Lambda,[p]}\)
is naturally identified with the associated graded vector space in the
regular case \(p=0\), by identifying PBW symbols having the same labels.
In particular,
\[
 \operatorname{ch}_q
 \operatorname{gr}M_{\mathrm{F}\oplus\mathrm{NS}}^{\Lambda,[p]}
 =
 \chi_{\mathrm F}(q)^2\chi_{\mathrm B}(q),
\]
where
\[
 \chi_{\mathrm F}(q)
 =
 \prod_{n=1}^{\infty}\left(1+q^{n-\frac12}\right),
 \qquad
 \chi_{\mathrm B}(q)
 =
 \prod_{n=1}^{\infty}\frac{1}{1-q^n}.
\]

\item
For every \(2m\in\mathbb Z\),
\[
 \deg |P,m\rangle=2m^2.
\]

\item
For \(i=1,2\) and \(n\geq1\),
\[
 L_{p-n}^{(i)}\mathcal U_{\leq d}^{[p]}
 \subset
 \mathcal U_{\leq d+n}^{[p]}.
\]
Consequently, if
\[
 \lambda=(\lambda_1,\ldots,\lambda_r),
 \qquad
 \mu=(\mu_1,\ldots,\mu_s)
\]
are partitions and
\[
 L_{p-\lambda}^{(1)}
 =
 L_{p-\lambda_1}^{(1)}\cdots L_{p-\lambda_r}^{(1)},
 \qquad
 L_{p-\mu}^{(2)}
 =
 L_{p-\mu_1}^{(2)}\cdots L_{p-\mu_s}^{(2)},
\]
then
\[
 \deg\left(
 L_{p-\lambda}^{(1)}
 L_{p-\mu}^{(2)}
 |P,m\rangle
 \right)
 \leq
 2m^2+|\lambda|+|\mu|.
\]
\end{enumerate}
\end{lem}

\begin{proof}
The PBW basis of
\(M_{\mathrm{F}\oplus\mathrm{NS}}^{\Lambda,[p]}\)
consists of the vectors
\[
 f_{-\mu}F_{-\lambda}|P,0\rangle,
\]
where the \(\mu_i\) are distinct positive half-integers and the
multiplicity of each odd part of \(\lambda\) is at most one.
The set of labels \((\mu,\lambda)\), as well as the degree
\[
 |\mu|+\frac{|\lambda|}{2},
\]
is independent of \(p\).  Identifying PBW symbols with the same labels
therefore gives a natural graded-vector-space isomorphism with the
corresponding associated graded module for \(p=0\).  The even parts of
\(\lambda\) give one bosonic generator in each positive integral degree,
while the odd parts of \(\lambda\) and the generators \(f_{-r}\) give two
fermionic generators in each positive half-integral degree.  This proves
the character formula in part~(1).

We next prove part~(2).  Set $m\neq 0$. By Lemma~\ref{lem:free-field-filtration}, $\deg |P,m\rangle\leq 2m^2$. 
The expansion of $|P,m\rangle$ contains the nonzero PBW monomial
\[
 f_{-\left(2|m|-1/2\right)}\cdots f_{-3/2}f_{-1/2}
 |P,0\rangle
\]
with coefficient \(1\). Therefore, no cancellation can remove its
highest-degree component, and
\[
 \deg |P,m\rangle
 =
 \sum_{i=1}^{2|m|}\left(i-1/2\right)
 =
 \frac{(2m)^2}{2}
 =
 2m^2.
\]
The claim is immediate for \(m=0\).

Finally, the defining formulas for \(L_{p-n}^{(i)}\) show that each of
their three terms,
\[
 L_{p-n},
 \qquad
 \sum_r r:f_{p-n-r}f_r:,
 \qquad
 \sum_r f_{p-n-r}G_r,
\]
raises the filtration degree by at most \(n\).  Hence
\[
 L_{p-n}^{(i)}\mathcal U_{\leq d}^{[p]}
 \subset
 \mathcal U_{\leq d+n}^{[p]},
\]
which proves part~(3).
\end{proof}

\begin{thm}\label{thm_decomposition_module}
Assume that \(P_p\neq 0\). 
Then, as a module over \(\mathrm{Vir}\oplus\mathrm{Vir}\),
\[
 M_{\mathrm{F}\oplus\mathrm{NS}}^{\Lambda,[p]}
 \cong
 \bigoplus_{2m\in\mathbb Z}
 M_{\mathrm{Vir}\oplus\mathrm{Vir}}^{\Lambda,m,[p]}.
\]
The irregular vector of
\(M_{\mathrm{Vir}\oplus\mathrm{Vir}}^{\Lambda,m,[p]}\)
is
\[
 |P,m\rangle
 =
 |\Lambda^{(m,1)}\rangle\otimes
 |\Lambda^{(m,2)}\rangle,
\]
where the weights are given by
\eqref{eq_weights_Vir_decomp}.
\end{thm}

\begin{proof}
For every \(2m\in\mathbb Z\), let
\[
 \mathcal M_m
 =
 U(\mathrm{Vir}\oplus\mathrm{Vir})|P,m\rangle.
\]
By the preceding discussion,
\[
 \mathcal M_m
 \cong
 M_{\mathrm{Vir}\oplus\mathrm{Vir}}^{\Lambda,m,[p]}.
\]
We first show that the sum of the submodules \(\mathcal M_m\) is direct. 
Set
\[
 \lambda_m
 =
 \Lambda_p^{(m,1)}
 =
 \frac{b^{-1}}{b^{-1}-b}\Lambda_p
 -
 \frac{2im}{b^{-1}-b}P_p.
\]
These numbers are mutually distinct because \(P_p\neq0\) and
\(b^2\neq1\).

On \(\mathcal M_m\), the operator
\[
 L_p^{(1)}-\lambda_m
\]
is locally nilpotent because of the Virasoro version of Lemma \ref{lem:NS-filtration} 

Thus \(\mathcal M_m\) is contained in the generalized
\(\lambda_m\)-eigenspace of \(L_p^{(1)}\).  Since generalized eigenspaces
corresponding to distinct eigenvalues have zero intersection, we obtain 
\[
 \sum_{2m\in\mathbb Z}\mathcal M_m
 =
 \bigoplus_{2m\in\mathbb Z}\mathcal M_m.
\]

It remains to prove that this direct sum exhausts
\(M_{\mathrm{F}\oplus\mathrm{NS}}^{\Lambda,[p]}\). 
Equip
\(M_{\mathrm{Vir}\oplus\mathrm{Vir}}^{\Lambda,m,[p]}\)
with the shifted PBW filtration defined by
\[
 \deg\left(
 L_{p-\lambda}^{(1)}
 L_{p-\mu}^{(2)}
 \bigl(
 |\Lambda^{(m,1)}\rangle
 \otimes
 |\Lambda^{(m,2)}\rangle
 \bigr)
 \right)
 =
 2m^2+|\lambda|+|\mu|.
\]
The preceding lemma shows that the homomorphism
\(\phi=\bigoplus_{2m\in\mathbb Z}\phi_m\),
where 
\[
 \phi_m:
 M_{\mathrm{Vir}\oplus\mathrm{Vir}}^{\Lambda,m,[p]}
 \longrightarrow
 M_{\mathrm{F}\oplus\mathrm{NS}}^{\Lambda,[p]},
 \qquad
 \phi_m\left(
 |\Lambda^{(m,1)}\rangle\otimes|\Lambda^{(m,2)}\rangle
 \right)
 =
 |P,m\rangle,
\]
is a filtered injection.

The graded character of the left-hand side is
\[
 \sum_{2m\in\mathbb Z}
 q^{2m^2}\chi_{\mathrm B}(q)^2.
\]
On the other hand, Lemma~\ref{lem_graded_decomposition} gives
\[
 \operatorname{ch}_q\operatorname{gr}
 M_{\mathrm{F}\oplus\mathrm{NS}}^{\Lambda,[p]}
 =
 \chi_{\mathrm F}(q)^2\chi_{\mathrm B}(q).
\]
The Jacobi triple product identity
\[
 \prod_{n=1}^{\infty}
 (1-q^n)\left(1+q^{n-\frac12}\right)^2
 =
 \sum_{k\in\mathbb Z}q^{k^2/2}
\]
implies
\[
 \chi_{\mathrm F}(q)^2\chi_{\mathrm B}(q)
 =
 \sum_{k\in\mathbb Z}
 q^{k^2/2}\chi_{\mathrm B}(q)^2
 =
 \sum_{2m\in\mathbb Z}
 q^{2m^2}\chi_{\mathrm B}(q)^2.
\]
Hence the source and target of the filtered injection have the same
dimension in every filtered degree.  The homomorphism is therefore
surjective in every filtered degree, and consequently
\[
 M_{\mathrm{F}\oplus\mathrm{NS}}^{\Lambda,[p]}
 =
 \bigoplus_{2m\in\mathbb Z}\mathcal M_m.
\]
This proves the theorem.
\end{proof}

\subsection{Decomposition of the irregular vertex operators}

We consider the operator $1\otimes\Phi_{\Lambda',\Lambda}^\Delta(z)$ from
$M^{\Lambda,[p]}_{\mathrm{F}\oplus\mathrm{NS}}$ to
$M^{\Lambda',[p]}_{\mathrm{F}\oplus\mathrm{NS}}$.  By
Theorem~\ref{thm_decomposition_module}, the image of $|P,m\rangle$
decomposes as
\begin{equation}\label{eq_wmz}
  1\otimes\Phi_{\Lambda',\Lambda}^\Delta(z)|P,m\rangle
  =\sum_{2n\in\mathbb Z}w_{mn}(z),
  \qquad
  w_{mn}(z)\in M_{\Vir\oplus\Vir}^{\Lambda',n,[p]}.
\end{equation}
We prove that each component \(w_{mn}(z)\) is a scalar multiple of a
tensor product of two Virasoro irregular vertex operators.  Let
$V_{{\Lambda'},\Lambda}^{\Delta}(z)$ denote the Virasoro irregular vertex
operator from $M_{\Vir}^{\Lambda,[p]}$ to
$M_{\Vir}^{\Lambda',[p]}$ defined by 
\begin{align*}
  &\left[L_n,V_{{\Lambda'},\Lambda}^{\Delta}(z)\right]=z^n\left(z\frac{\partial}{\partial z}+(n+1)\Delta\right)V_{{\Lambda'},\Lambda}^{\Delta}(z),
  \\
  &V_{\Lambda',\Lambda}^{\Delta}(z)|\Lambda\rangle =z^\alpha \exp\left(\sum_{i=1}^p\frac{\beta_i}{z^i}\right)\sum_{k=0}^\infty v_kz^k,
  \quad v_0=|\Lambda'\rangle.
\end{align*}
The coefficients $v_k$ satisfy relations similar to \eqref{eq_def_rel_1}:
\begin{align}
  &\left(L_\ell-{\Lambda'}_\ell\right)v_{k}=\left( L_p-\Lambda_p\right)
  v_{k-\ell+p}+(\ell-p)\Delta v_{k-\ell}\quad (\ell>p),\label{eq_Vir_v__p}
  \\
  &\left(L_p-{\Lambda'}_p\right)v_k=-\sum_{i=1}^{p-1}i\beta_i 
  v_{k-p+i}+(\alpha+(p+1)\Delta+k-p)v_{k-p}, \label{eq_Vir_v_p}
\end{align}
These relations determine $\alpha$, $\beta_1,\ldots,\beta_{p-1}$,
$\Lambda'$, and $v_k$ uniquely in terms of $\Lambda$, $\Delta$, and
$\beta_p$ \cite{Nagoya2015}.  In particular,
\[
 \Lambda'_\ell=\Lambda_\ell-\delta_{\ell,p}p\beta_p.
\] 

Set
\begin{equation*}
 \Delta^{(1)}=\frac{b^{-1}}{b^{-1}-b}\Delta,\quad 
 \Delta^{(2)}=\frac{b}{b-b^{-1}}\Delta.   
 \end{equation*} 
\begin{thm}\label{thm_IVO_decomposition}
	For every $2m,2n\in\Z$, the component of
 $1\otimes \Phi_{\Lambda',\Lambda}^\Delta(z)$ 
 mapping the $m$-th source summand to the $n$-th target summand is a scalar multiple of the tensor product of Virasoro irregular vertex operators:  
 \begin{equation*}
	1\otimes \Phi_{\Lambda', \Lambda}^\Delta(z)
	|_{M_{\mathrm{Vir}\oplus\mathrm{Vir}}^{\Lambda,m,[p]}\to 
	M_{\mathrm{Vir}\oplus\mathrm{Vir}}^{\Lambda',n,[p]}}
	=\mathsf a_{mn}\left(P,\beta_p,\Delta,b\right)
	V_{{\Lambda'}^{(n,1)},\Lambda^{(m,1)}}^{\Delta^{(1)}}(z)\otimes V_{{\Lambda'}^{(n,2)},\Lambda^{(m,2)}}^{\Delta^{(2)}}(z),
 \end{equation*}
 where \(\mathsf a_{mn}\) is a scalar depending on the parameters $P, \beta_p, \Delta,b$, and 
 $|_{M_{\mathrm{Vir}\oplus\mathrm{Vir}}^{\Lambda,m,[p]}\to 
	M_{\mathrm{Vir}\oplus\mathrm{Vir}}^{\Lambda',n,[p]}}$ denotes the corresponding component map 
  between these submodules. The scalar \(\mathsf a_{mn}\) is determined by the leading component of
\((1\otimes\Phi_{\Lambda',\Lambda}^{\Delta}(z))|P,m\rangle\) in the
summand \(M_{\Vir\oplus\Vir}^{\Lambda',n,[p]}\).
\end{thm}
\begin{proof}
Let $|P,m\rangle$ and $|P',n\rangle$ denote the irregular vectors of
$M_{\Vir\oplus\Vir}^{\Lambda,m,[p]}$ and
$M_{\Vir\oplus\Vir}^{\Lambda',n,[p]}$, respectively. By \eqref{eq_wmz}, it is
enough to identify each component $w_{mn}(z)$.  We shall prove that
\[
 w_{mn}(z)
 =\mathsf a_{mn}(P,\beta_p,\Delta,b)
 \bigl(
 V_{{\Lambda'}^{(n,1)},\Lambda^{(m,1)}}^{\Delta^{(1)}}(z)
 \otimes
 V_{{\Lambda'}^{(n,2)},\Lambda^{(m,2)}}^{\Delta^{(2)}}(z)
 \bigr)|P,m\rangle.
\]
If \(w_{mn}(z)=0\), set \(\mathsf a_{mn}=0\). We henceforth assume that \(w_{mn}(z)\neq0\).

\paragraph{Step 1. Recursion relations.}
The commutation relations of $\Phi_{\Lambda',\Lambda}^{\Delta}(z)$ imply that,
for every integer $q$,
\begin{align*}
 &\left(
   \sum_r f_{q-r}G_r
   -z^{q-p}\sum_r f_{p-r}G_r
  \right)
  \bigl(1\otimes\Phi_{\Lambda',\Lambda}^{\Delta}(z)\bigr)
 \\
 &\quad=
 \sum_r f_{q-r}
 \Bigl\{
   z^{r+\frac12}\Psi_{\Lambda',\Lambda}^{\Delta}(z)
   -z^{q-p}z^{r-q+p+\frac12}
    \Psi_{\Lambda',\Lambda}^{\Delta}(z)
   +\Phi_{\Lambda',\Lambda}^{\Delta}(z)
    (G_r-z^{q-p}G_{r-q+p})
 \Bigr\}
 \\
 &\quad=
 \bigl(1\otimes\Phi_{\Lambda',\Lambda}^{\Delta}(z)\bigr)
 \left(
   \sum_r f_{q-r}G_r
   -z^{q-p}\sum_r f_{p-r}G_r
 \right).
\end{align*}
Consequently, for any integers $n_1,\ldots,n_\ell$,
\begin{align*}
 &\bigl(1\otimes\Phi_{\Lambda',\Lambda}^{\Delta}(z)\bigr)
 \prod_{i=1}^{\ell}
 \left(
   L_{n_i}^{(1)}-z^{n_i-p}L_p^{(1)}
   +(n_i-p)\Delta^{(1)}z^{n_i}
 \right)|P,m\rangle
 \\
 &\qquad=
 \prod_{i=1}^{\ell}
 \left(L_{n_i}^{(1)}-z^{n_i-p}L_p^{(1)}\right)
 \bigl(1\otimes\Phi_{\Lambda',\Lambda}^{\Delta}(z)\bigr)|P,m\rangle.
\end{align*}

By the defining expansion~\eqref{eq:phi_irr}, the component $w_{mn}(z)$ can be written as
\[
 w_{mn}(z)
 =z^{\alpha+\alpha_{mn}}
  \exp\left(\sum_{i=1}^{p}\frac{\beta_i}{z^i}\right)
  \sum_{2k\in\mathbb Z_{\geq0}}w_k^{(mn)}z^k,
\]
where $2\alpha_{mn}\in\mathbb Z$ and $w_0^{(mn)}\neq0$.  The parameters
$\alpha$ and $\beta_i$ are those of
$\Phi_{\Lambda',\Lambda}^{\Delta}(z)$.
Using the embedding $\Vir\oplus\Vir\subset \mathrm{F}\oplus\mathrm{NS}$, we obtain, for
$\ell>p$ and $i=1,2$,
\[
 L_{\ell}^{(i)}
 \bigl(1\otimes\Phi_{\Lambda',\Lambda}^{\Delta}(z)\bigr)|P,m\rangle
 =\left\{
   \bigl(L_p^{(i)}-\Lambda_p^{(m,i)}\bigr)z^{\ell-p}
   +(\ell-p)\Delta^{(i)}z^{\ell}
   +\Lambda_{\ell}^{(m,i)}
  \right\}
 \bigl(1\otimes\Phi_{\Lambda',\Lambda}^{\Delta}(z)\bigr)|P,m\rangle.
\]
We also have
\[
 \bigl(L_p^{(1)}+L_p^{(2)}\bigr)
 \bigl(1\otimes\Phi_{\Lambda',\Lambda}^{\Delta}(z)\bigr)|P,m\rangle
 =\left\{
   z^p\left(z\frac{\partial}{\partial z}+(p+1)\Delta\right)
   +\Lambda_p
  \right\}
 \bigl(1\otimes\Phi_{\Lambda',\Lambda}^{\Delta}(z)\bigr)|P,m\rangle.
\]
Comparing coefficients in this expansion gives
\begin{equation}\label{eq_VirL__p_w}
 \bigl(L_{\ell}^{(i)}-{\Lambda'}_{\ell}^{(n,i)}\bigr)w_k^{(mn)}
 =\bigl(L_p^{(i)}-\Lambda_p^{(m,i)}\bigr)w_{k-\ell+p}^{(mn)}
  +(\ell-p)\Delta^{(i)}w_{k-\ell}^{(mn)},
 \qquad \ell>p,
\end{equation}
and
\begin{equation}\label{eq_VirL_p_w}
  \bigl(
   L_p^{(1)}-{\Lambda'}_p^{(n,1)}
   +L_p^{(2)}-{\Lambda'}_p^{(n,2)}
  \bigr)w_k^{(mn)}
 =
 -\sum_{i=1}^{p-1}i\beta_iw_{k-p+i}^{(mn)}
 +B_kw_{k-p}^{(mn)}, 
\end{equation}
where $B_k=\alpha+\alpha_{mn}+(p+1)\Delta+k-p$.

Equation~\eqref{eq_VirL__p_w} coincides with the higher-mode
recursion~\eqref{eq_Vir_v__p} for each Virasoro factor, whereas
equation~\eqref{eq_VirL_p_w} gives the combined \(L_p\)-relation.

\paragraph{Step 2. Initial coefficients and the induction statement.}
Setting $k=0$ in~\eqref{eq_VirL__p_w}, we find that
\[
 \bigl(L_{\ell}^{(i)}-{\Lambda'}_{\ell}^{(n,i)}\bigr)w_0^{(mn)}=0,
 \qquad \ell\geq p+1.
\]
It follows that
\[
 w_0^{(mn)}
 =\mathsf a_{mn}(P,\beta_p,\Delta,b)|P',n\rangle
\]
for some scalar $\mathsf a_{mn}(P,\beta_p,\Delta,b)$.  Similarly,
\[
 w_{1/2}^{(mn)}
 =\mathsf b_{mn}(P,\beta_p,\Delta,b)|P',n\rangle
\]
for some scalar $\mathsf b_{mn}(P,\beta_p,\Delta,b)$.
We now prove, by induction on $k$, that
\begin{equation}\label{eq_w_vv}
 w_k^{(mn)}
 =\mathsf a_{mn}(P,\beta_p,\Delta,b)
  \sum_{j=0}^{k}v_j^{(1)}\otimes v_{k-j}^{(2)},
 \qquad
 v_0^{(i)}=|{\Lambda'}^{(n,i)}\rangle,
 \quad
 v_j^{(i)}\in M_{\Vir}^{{\Lambda'}^{(n,i)},[p]},
\end{equation}
where the vectors $v_j^{(i)}$ satisfy~\eqref{eq_Vir_v__p} and~\eqref{eq_Vir_v_p}.  At the same time,
we determine the parameters 
$\beta_j^{(mn,i)}$ and $\alpha^{(mn,i)}$ of the two Virasoro irregular
vertex operators. This also forces $w_{k+1/2}^{(mn)}=0$ for all $k\in\mathbb Z_{\geq0}$ because of the uniqueness of Virasoro irregular vertex operators \cite{Nagoya2015}. Indeed, a nonzero half-integral subsequence would define a Virasoro
irregular vertex operator with the same source and target weights but
with its leading exponent shifted by \(1/2\), contradicting the
uniqueness of the leading exponent. 
We use the convention $v_j^{(i)}=0$ for $j<0$.

For $k=1$, equation~\eqref{eq_VirL__p_w} gives
\[
 \bigl(L_{\ell}^{(i)}-{\Lambda'}_{\ell}^{(n,i)}\bigr)w_1^{(mn)}=0,
 \qquad \ell\geq p+2.
\]
Hence
\[
 w_1^{(mn)}
 =\mathsf a_{mn}(P,\beta_p,\Delta,b)
 \bigl(v_1^{(1)}\otimes v_0^{(2)}
      +v_0^{(1)}\otimes v_1^{(2)}\bigr),
\]
where
\[
 v_1^{(i)}
 =c_{\emptyset}^{(i)}v_0^{(i)}
  +c_1^{(i)}L_{p-1}^{(i)}v_0^{(i)}.
\]
Moreover,
\[
 \bigl(L_{p+1}^{(i)}-{\Lambda'}_{p+1}^{(n,i)}\bigr)w_1^{(mn)}
 =-p\beta_p^{(mn,i)}w_0^{(mn)},
\]
and
\[
 \bigl(
  L_p^{(1)}-{\Lambda'}_p^{(n,1)}
  +L_p^{(2)}-{\Lambda'}_p^{(n,2)}
 \bigr)w_1^{(mn)}
 =-(p-1)\beta_{p-1}w_0^{(mn)},
\]
where
\[
 p\beta_p^{(mn,i)}
 =\Lambda_p^{(m,i)}-{\Lambda'}_p^{(n,i)}.
\]
Since
$\beta_p^{(mn,1)}+\beta_p^{(mn,2)}=\beta_p$, we may choose
$\beta_{p-1}^{(mn,1)}$ and $\beta_{p-1}^{(mn,2)}$ so that
\[
 \beta_{p-1}^{(mn,1)}+\beta_{p-1}^{(mn,2)}=\beta_{p-1}
\]
and
\[
 \bigl(L_p^{(i)}-{\Lambda'}_p^{(n,i)}\bigr)v_1^{(i)}
 =-(p-1)\beta_{p-1}^{(mn,i)}v_0^{(i)}.
\]
Thus $v_1^{(i)}$ satisfies~\eqref{eq_Vir_v__p} and~\eqref{eq_Vir_v_p}.

If \(p>1\), the combined \(L_p\)-relation at level \(1\)
determines \(\beta_{p-1}^{(mn,1)}\) and
\(\beta_{p-1}^{(mn,2)}\) as above.

If \(p=1\), level \(1\) is the case \(N=p\).  We choose
\[
 \alpha^{(mn,1)}+\alpha^{(mn,2)}
 =
 \alpha+\alpha_{mn}
\]
and impose the constant-term conditions corresponding to the
\(N=p\) case below.

\paragraph{Step 3. Induction step.}
Assume that~\eqref{eq_w_vv} holds for $k=1,\ldots,N-1$.  Let
$\beta_j^{(mn,i)}$ $(j=1,\ldots,p-1)$ and $\alpha^{(mn,i)}$ denote the
parameters of the two Virasoro irregular vertex operators.  If $N\leq p$,
we assume that
\[
 \beta_j^{(mn,1)}+\beta_j^{(mn,2)}=\beta_j,
 \qquad j=p-1,\ldots,p-(N-1).
\]
If $N>p$, we assume that
\[
 \beta_j^{(mn,1)}+\beta_j^{(mn,2)}=\beta_j,
 \qquad j=1,\ldots,p-1,
\]
and
\[
 \alpha^{(mn,1)}+\alpha^{(mn,2)}=\alpha+\alpha_{mn}.
\]
Put
\[
B_{j,\ell}^{(i)}=\alpha^{(mn,i)}+(\ell+1)\Delta^{(i)}+j-\ell,\qquad i=1,2,
\]
and set
\[
 R_N
 =
 \mathsf a_{mn}^{-1}w_N^{(mn)}
 -\sum_{j=1}^{N-1}v_j^{(1)}\otimes v_{N-j}^{(2)}.
\]
For $\ell\geq p+1$, equations~\eqref{eq_VirL__p_w}, \eqref{eq_VirL_p_w},
\eqref{eq_w_vv}, and the induction hypothesis give
\begin{align*}
 &\bigl(L_{\ell}^{(1)}-{\Lambda'}_{\ell}^{(n,1)}\bigr)R_N
 \\
 &={}
 \sum_{j=0}^{N-\ell+p}
 \bigl(-p\beta_p^{(mn,1)}+L_p^{(1)}-{\Lambda'}_p^{(n,1)}\bigr)
 v_j^{(1)}\otimes v_{N-\ell+p-j}^{(2)}
 \\
 &\quad-
 \sum_{j=1}^{N-1}
 \bigl(L_{\ell}^{(1)}-{\Lambda'}_{\ell}^{(n,1)}\bigr)
 v_j^{(1)}\otimes v_{N-j}^{(2)}
 \\
 &\quad+
 \mathsf a_{mn}^{-1}(\ell-p)\Delta^{(1)}w_{N-\ell}^{(mn)}
 \\
 &={}
 -p\beta_p^{(mn,1)}
 \sum_{j=0}^{N-\ell+p}
 v_j^{(1)}\otimes v_{N-\ell+p-j}^{(2)}
 \\
 &\quad+
 \sum_{j=0}^{N-\ell+p}
 \left(
  -\sum_{i=1}^{p-1}i\beta_i^{(mn,1)}v_{j-p+i}^{(1)}
  +B_{j,p}^{(1)}v_{j-p}^{(1)}
 \right)
 \otimes v_{N-\ell+p-j}^{(2)}
 \\
 &\quad-
 \sum_{j=1}^{N-1}
 \left(
  -\sum_{i=1}^{p}i\beta_i^{(mn,1)}v_{j-\ell+i}^{(1)}
  +B_{j,\ell}^{(1)}v_{j-\ell}^{(1)}
 \right)
 \otimes v_{N-j}^{(2)}
 \\
 &\quad+
 \mathsf a_{mn}^{-1}(\ell-p)\Delta^{(1)}w_{N-\ell}^{(mn)}
 \\
 &={}
 \left(
  -\sum_{i=1}^{p}i\beta_i^{(mn,1)}v_{N-\ell+i}^{(1)}
  +B_{N,\ell}^{(1)}v_{N-\ell}^{(1)}
 \right)
 \\
 &\qquad\otimes v_0^{(2)}.
\end{align*}
Applying the Virasoro analogue of Corollary~\ref{cor_n_kill} to the two tensor factors,
we obtain
\[
 w_N^{(mn)}
 =\mathsf a_{mn}(P,\beta_p,\Delta,b)
 \left(
  \bar v_N^{(1)}\otimes v_0^{(2)}
  +v_0^{(1)}\otimes\bar v_N^{(2)}
  +\sum_{j=1}^{N-1}v_j^{(1)}\otimes v_{N-j}^{(2)}
 \right),
\]
where $\bar v_N^{(i)}\in M_{\Vir}^{{\Lambda'}^{(n,i)},[p]}$ satisfies
\[
 \bigl(L_{\ell}^{(i)}-{\Lambda'}_{\ell}^{(n,i)}\bigr)\bar v_N^{(i)}
 =-\sum_{j=1}^{p}j\beta_j^{(mn,i)}v_{N-\ell+j}^{(i)}
  +B_{N,\ell}^{(i)}
    v_{N-\ell}^{(i)}
\]
for all $\ell\geq p+1$.

\paragraph{The \texorpdfstring{$L_p$}{L_p}-relation.}
It remains to impose the $L_p$-relation.  Set
\begin{align*}
 T_N
 &:={}
 \bigl(
  L_p^{(1)}-{\Lambda'}_p^{(n,1)}
  +L_p^{(2)}-{\Lambda'}_p^{(n,2)}
 \bigr)
 \bigl(
  \bar v_N^{(1)}\otimes v_0^{(2)}
  +v_0^{(1)}\otimes\bar v_N^{(2)}
 \bigr)
 \\
 &={}
 \bigl(
  L_p^{(1)}-{\Lambda'}_p^{(n,1)}
  +L_p^{(2)}-{\Lambda'}_p^{(n,2)}
 \bigr)R_N
 \\
 &={}
 -\mathsf a_{mn}^{-1}\sum_{i=1}^{p-1}i\beta_iw_{N-p+i}^{(mn)}
 +B_N
   \mathsf a_{mn}^{-1}w_{N-p}^{(mn)}
 \\
 &\quad+
 \sum_{j=1}^{N-1}
 \left(
  \sum_{i=1}^{p-1}i\beta_i^{(mn,1)}v_{j-p+i}^{(1)}
  -B_{j,p}^{(1)}v_{j-p}^{(1)}
 \right)\otimes v_{N-j}^{(2)}
 \\
 &\quad+
 \sum_{j=1}^{N-1}v_j^{(1)}\otimes
 \left(
  \sum_{i=1}^{p-1}i\beta_i^{(mn,2)}v_{N-j-p+i}^{(2)}
  -B_{N-j,p}^{(2)}
   v_{N-j-p}^{(2)}
 \right).
\end{align*}
We now distinguish three cases.  If $N<p$, then
\[
 \begin{aligned}
 T_N={}&-
 \sum_{j=p-N+1}^{p-1}
 \left(
  j\beta_j^{(mn,2)}v_0^{(1)}\otimes v_{N-p+j}^{(2)}
  +j\beta_j^{(mn,1)}v_{N-p+j}^{(1)}\otimes v_0^{(2)}
 \right)
 \\
 &\quad-(p-N)\beta_{p-N}v_0^{(1)}\otimes v_0^{(2)}.
 \end{aligned}
\]
If $N=p$, then
\[
 \begin{aligned}
 T_N={}&-
 \sum_{j=1}^{p-1}
 \left(
  j\beta_j^{(mn,2)}v_0^{(1)}\otimes v_{N-p+j}^{(2)}
  +j\beta_j^{(mn,1)}v_{N-p+j}^{(1)}\otimes v_0^{(2)}
 \right)
 +B_p
 v_0^{(1)}\otimes v_0^{(2)}.
 \end{aligned}
\]
Finally, if $N>p$, then
\[
 \begin{aligned}
 T_N={}&-
 \sum_{j=1}^{p-1}
 \left(
  j\beta_j^{(mn,2)}v_0^{(1)}\otimes v_{N-p+j}^{(2)}
  +j\beta_j^{(mn,1)}v_{N-p+j}^{(1)}\otimes v_0^{(2)}
 \right)
 \\
 &\quad+B_{N,p}^{(2)}
  v_0^{(1)}\otimes v_{N-p}^{(2)}
 +B_{N,p}^{(1)}
  v_{N-p}^{(1)}\otimes v_0^{(2)}.
 \end{aligned}
\]

For $u\in M_{\Vir}^{{\Lambda'}^{(n,i)},[p]}$, let $\{u\}$ denote its
constant term.  When $N<p$, choose $\beta_{p-N}^{(mn,1)}$ and
$\beta_{p-N}^{(mn,2)}$ so that
\[
 \beta_{p-N}^{(mn,1)}+\beta_{p-N}^{(mn,2)}=\beta_{p-N}
\]
and impose, for $i=1,2$,
\[
 \left\{
  \bigl(L_p^{(i)}-{\Lambda'}_p^{(n,i)}\bigr)\bar v_N^{(i)}
  +\sum_{j=p-N+1}^{p-1}j\beta_j^{(mn,i)}v_{N-p+j}^{(i)}
  +(p-N)\beta_{p-N}^{(mn,i)}v_0^{(i)}
 \right\}=0.
\]
When $N=p$, choose $\alpha^{(mn,1)}$ and $\alpha^{(mn,2)}$ so that
\[
 \alpha^{(mn,1)}+\alpha^{(mn,2)}=\alpha+\alpha_{mn}
\]
and impose
\[
 \left\{
  \bigl(L_p^{(i)}-{\Lambda'}_p^{(n,i)}\bigr)\bar v_N^{(i)}
  +\sum_{j=1}^{p-1}j\beta_j^{(mn,i)}v_{N-p+j}^{(i)}
  -B_{p,p}^{(i)}v_0^{(i)}
 \right\}=0.
\]
When $N>p$, impose
\[
 \left\{
  \bigl(L_p^{(i)}-{\Lambda'}_p^{(n,i)}\bigr)\bar v_N^{(i)}
  +\sum_{j=1}^{p-1}j\beta_j^{(mn,i)}v_{N-p+j}^{(i)}
  -B_{N,p}^{(i)}v_{N-p}^{(i)}
 \right\}=0.
\]
These are the constant-term conditions associated with~\eqref{eq_Vir_v_p}.  As in the
proof of Theorem~\ref{thm_IVO}, the higher-mode relations together with these
constant-term conditions determine the vectors $v_j^{(i)}$ uniquely.
Therefore $\bar v_N^{(i)}=v_N^{(i)}$, and~\eqref{eq_w_vv} holds at level $N$.
The induction is complete, and hence so is the proof.
\end{proof}

\subsection{On the generating function \texorpdfstring{$f(z)$}{f(z)} of free fermions}

Let $f(z)$ be defined as 
\begin{equation*}
	f(z)=\sum_{r\in\Z+1/2}f_r z^{-r-\frac12}. 
\end{equation*}

The generating function $f(z)$ of the free fermions decomposes into a sum of tensor products of degenerate Virasoro fields. 
\begin{prop}\label{prop_fz}
We have 
	\begin{align*}
f(z)|P,n\rangle=&x_n^+V_{{\Lambda}^{(n+1/2,1)},\Lambda^{(n,1)}}^{\Delta_f^{(1)}}(z)
	\otimes V_{{\Lambda}^{(n+1/2,2)},\Lambda^{(n,2)}}^{\Delta_f^{(2)}}(z)|P,n\rangle
	\\&+x_n^-V_{{\Lambda}^{(n-1/2,1)},\Lambda^{(n,1)}}^{\Delta_f^{(1)}}(z)
	\otimes V_{{\Lambda}^{(n-1/2,2)},\Lambda^{(n,2)}}^{\Delta_f^{(2)}}(z)|P,n\rangle,  
	\end{align*}
	where 
\begin{align*}
&x_n^+=\begin{cases}
    \frac12,& n\geq 0,
    \\
    1, & n<0,
\end{cases}\quad x_n^-=\begin{cases}
    1,& n> 0,
    \\
    \frac12, & n\leq0,
\end{cases}
\\
  &\Delta_f^{(1)}=\Delta_{1,2}(b^{(1)})=-\frac{b^{-1}+2b}{2(b^{-1}-b)},
  \quad 
  \Delta_f^{(2)}=\Delta_{2,1}(b^{(2)})=-\frac{b+2b^{-1}}{2(b-b^{-1})}.
  \end{align*}
 	\end{prop}

\begin{proof}
Using the embedding, we obtain 
\begin{align*}
  &L_\ell^{(i)}f(z)|P,n\rangle=
  \left\{ \left(L_p^{(i)}-\Lambda_p^{(n,i)}\right)z^{\ell-p}
  +(\ell-p)\Delta_f^{(i)}z^\ell+\Lambda_\ell^{(n,i)}\right\}f(z)
  |P,n\rangle\quad (i=1,2,\, \ell>p),
  \\
  &\left(L_p^{(1)}+L_p^{(2)}\right)f(z)|P,n\rangle=
  \left\{z^p\left(z\frac{\partial}{\partial z}
  +\frac{1}{2}(p+1)\right)+\Lambda_p\right\}f(z)|P,n\rangle. 
\end{align*}
Hence, arguing as in the proof of Theorem~\ref{thm_IVO_decomposition}, we have 
\begin{equation*}
  f(z)|P,n\rangle=\sum_{2m\in\Z}x_{mn}
	V^{\Delta_f^{(1)}}_{\Lambda^{(m,1)},\Lambda^{(n,1)}}(z) 
  \otimes V^{\Delta_f^{(2)}}_{\Lambda^{(m,2)},\Lambda^{(n,2)}}(z)|P,n\rangle,  
\end{equation*}
where $x_{mn}$ is a scalar depending on $P$ and $b$.  Here $x_{mn}$
denotes the component from the $n$-th source summand to the $m$-th target
summand. 

A direct calculation shows that, for \(i=1,2\),
\begin{align*}
 &\left(
   \left(L_{-1}^{(i)}\right)^2f(z)
   -2L_{-1}^{(i)}f(z)L_{-1}^{(i)}
   +f(z)\left(L_{-1}^{(i)}\right)^2
  \right)|P,n\rangle
 \\
 &\qquad=
 -\left(b^{(i)}\right)^2
 \left(
  \sum_{k\leq-2}L_k^{(i)}z^{-k-2}f(z)
  +
  \sum_{k\geq-1}f(z)L_k^{(i)}z^{-k-2}
 \right)|P,n\rangle.
\end{align*}
The identity above is obtained solely from the Neveu--Schwarz commutation
relations.  In particular, its derivation does not use the action on
\(|P,n\rangle\).

For the component from the $n$-th source summand to the $m$-th target
summand, set
\[
 p\beta_p^{(nm,i)}
 =
 \Lambda_p^{(n,i)}-\Lambda_p^{(m,i)}.
\]
If this component is nonzero, the double commutator with
\(L_{-1}^{(i)}\) differentiates only the \(i\)-th Virasoro factor.
Comparing the coefficient of \(z^{-2p-2}\) gives
\[
 p^2\bigl(\beta_p^{(nm,i)}\bigr)^2
 +\bigl(b^{(i)}\bigr)^2\Lambda_{2p}^{(n,i)}
 =0.
\]
Using \eqref{eq_Lambda_P}, \eqref{eq_weights_Vir_decomp}, and the
definitions of \(b^{(i)}\), this equation is equivalent to
\[
 4(m-n)^2=1.
\]
Thus \(x_{mn}=0\) unless \(m=n\pm1/2\).

Therefore, we obtain 
\begin{align}
	f(z)|P,n\rangle=& x_{n+1/2,n} V_{{\Lambda}^{(n+1/2,1)},\Lambda^{(n,1)}}^{\Delta_f^{(1)}}(z)
	\otimes V_{{\Lambda}^{(n+1/2,2)},\Lambda^{(n,2)}}^{\Delta_f^{(2)}}(z)|P,n\rangle \nonumber 
			\\
			&+x_{n-1/2,n} V_{{\Lambda}^{(n-1/2,1)},\Lambda^{(n,1)}}^{\Delta_f^{(1)}}(z)
			\otimes V_{{\Lambda}^{(n-1/2,2)},\Lambda^{(n,2)}}^{\Delta_f^{(2)}}(z)|P,n\rangle\nonumber
			\\
			=&x_{n+1/2,n}z^{2n}(|P,n+1/2\rangle+O(z))+x_{n-1/2,n} z^{-2n}(|P,n-1/2\rangle+O(z)).  \label{eq:fz}
\end{align}

When $n=0$, looking at the constant term of \eqref{eq:fz} and using the relations \eqref{eq:psi-G-relation}, we obtain $x_{1/2,0}=x_{-1/2,0}=1/2$. 

We consider the case of $n>0$. 

Since 
\begin{align*}
  &|P,n\rangle=(f_{-2n+1/2}-i\psi_{-2n+1/2}^{(1)})\cdots (f_{-1/2}-i\psi_{-1/2}^{(1)})|P\rangle, 
\end{align*}
 we have 
  \begin{align*}
    f(z)|P,n\rangle=&\sum_{s\leq 2n-1/2}f_sz^{-s-1/2}|P,n\rangle. 
    \end{align*} 
  Hence, this gives $x_{n-1/2,n}=1$. 
	Moreover, the coefficient of $z^{2n}$ in the equation \eqref{eq:fz} reads as 
  \begin{equation*}
    |P,n+1/2\rangle+i\psi^{(1)}_{-2n-1/2}|P,n\rangle=x_{n+1/2,n}|P,n+1/2\rangle+v,  
  \end{equation*}
  where $v\in M_{\Vir\oplus\Vir}^{\Lambda,n-1/2,[p]}$. 
  Hence, \begin{equation}
 i\psi^{(1)}_{-2n-1/2}|P,n\rangle
 =
 \gamma|P,n+1/2\rangle+u_n\quad \left(\gamma\in\C, u_n\in
 M_{\mathrm{Vir}\oplus\mathrm{Vir}}^{\Lambda,n-1/2,[p]}\right).
 \label{eq:psi-leading-component}
\end{equation}

We determine \(\gamma\) by applying \(L_p^{(1)}\) and comparing the
degree-\(2(n+1/2)^2\) component in \(M_{\mathrm{Vir}\oplus\mathrm{Vir}}^{\Lambda,n+1/2,[p]}\). Put $s=-2n-1/2$. 
In the first free-field realization, one has
\begin{equation}
 \left[L_p^{(1)},\psi^{(1)}_{s}\right]
 =
 -\frac{b^{-1}}{b^{-1}-b}
 \left(s+\frac p2\right)\psi^{(1)}_{p+s}
 +\frac{1}{b^{-1}-b}\sum_{k\neq 0}
 a^{(1)}_k f_{p+s-k}
 +\frac{i}{b^{-1}-b}(-Qs-a^{(1)}_0)f_{p+s}.
 \label{eq:Lp-psi-commutator}
\end{equation}

In \eqref{eq:Lp-psi-commutator}, by Lemma~\ref{lem:free-field-filtration} the only term that can contribute to
filtration degree \(2(n+1/2)^2\) is the term with \(k=p\), namely
\[
 \frac{P_p}{b^{-1}-b}f_{s}|P,n\rangle.
\]

Consequently, after taking the degree-\(2(n+1/2)^2\) component, we obtain
\begin{equation}
 \left[L_p^{(1)},\psi^{(1)}_{s}\right]|P,n\rangle
 \equiv
 \frac{P_p}{b^{-1}-b}f_{s}|P,n\rangle.
 \label{eq:Lp-psi-top-comparison}
\end{equation}
Here and below, \(\equiv\) denotes equality of the
degree-\(2(n+1/2)^2\) components in the \(n+1/2\) target summand.

By \eqref{eq:psi-leading-component}, the left-hand side of
\eqref{eq:Lp-psi-top-comparison} becomes 
\begin{equation}
 -i\gamma\left(
 \Lambda_p^{(n+1/2,1)}
 -
 \Lambda_p^{(n,1)}
 \right)|P,n+1/2\rangle=
 -\gamma\frac{P_p}{b^{-1}-b}|P,n+1/2\rangle.
 \label{eq:Lambda-p-difference}
\end{equation}

On the other hand, the definition of \(|P,n+1/2\rangle\) gives
\[
 |P,n+1/2\rangle
 =
 \left(f_{s}-i\psi^{(1)}_{s}\right)|P,n\rangle.
\]
Using \eqref{eq:psi-leading-component}, we therefore have
\[
 f_{s}|P,n\rangle
 \equiv
 (1+\gamma)|P,n+1/2\rangle.
\]
Substituting this and \eqref{eq:Lambda-p-difference} into
\eqref{eq:Lp-psi-top-comparison}, we obtain
\[
 -\frac{P_p}{b^{-1}-b}\gamma
 =
 \frac{P_p}{b^{-1}-b}(1+\gamma).
\]
Since \(P_p\neq0\) and \(b^{-1}-b\neq0\), we have
\(\gamma=-1/2\).  Thus $x_{n+1/2,n}=1/2$.  The case $n<0$ is
analogous and gives $x_{n+1/2,n}=1$ and $x_{n-1/2,n}=1/2$.
This completes the proof. 
\end{proof}

\begin{re}
In the regular case, the expectation value of the free-fermion field
\(f(z)\) is used to determine the coefficients in the decomposition through
properties of the Gauss hypergeometric function \cite{BS}.  In the irregular
case, the corresponding pairings should likewise be computable explicitly.
More precisely, they are pairings of Virasoro irregular vectors embedded in
the Neveu--Schwarz and free-fermion Verma modules, and their expressions
should be obtainable from properties of the Kummer and Hermite functions.
These pairings, however, do not determine the coefficients
\(\mathsf a_{mn}\) themselves.  Rather, the full coefficients \(p_{mn}\)
have the factorized form
\[
p_{mn}=\mathsf a_{mn}\rho_{mn},
\]
where \(\rho_{mn}\) denotes the corresponding pairing factor.

 \end{re} 

\section{Pairings and bilinear operators for irregular conformal blocks}

In this section, we apply the decomposition theorem for irregular
Neveu--Schwarz vertex operators to matrix elements defined by suitable
pairings on $\mathrm{F}\oplus\mathrm{NS}$ modules.  We first introduce the
pairings defining the relevant $\mathrm{F}\oplus\mathrm{NS}$ conformal
blocks.  Under the decomposition theorem, these conformal blocks are
expressed as weighted sums of products of Virasoro irregular conformal blocks
of types $(0,0,1)$ and $(0,2)$.  Using the $H_n$-insertions together with the
decomposition theorem, we then derive bilinear equations for these weighted
sums.  Finally, we compare the resulting bilinear operators with those
appearing in the quantum Painlev\'e V and IV tau-function equations
\cite{BST}. 

Throughout this section, \(\beta\) denotes the highest irregular
parameter \(\beta_p\).

\subsection{A pairing for \texorpdfstring{$(0,0,1)$}{(0,0,1)}}

Let $M^{\Delta_0}_{\mathrm{F}\oplus\mathrm{NS}}$ be 
the tensor product of the Verma module $M_\mathrm{F}$ of 
the free-fermion algebra and the Verma module 
$M_{\mathrm{NS}}^{\Delta_0}$ of the Neveu--Schwarz algebra. 
We denote the highest weight vector of $M^{\Delta_0}_{\mathrm{F}\oplus\mathrm{NS}}$ 
by $| \Delta_0\rangle$. We regard the dual module as a right module and use the corresponding balanced pairing. 
A pairing between $M^{\Delta_0}_{\mathrm{F}\oplus\mathrm{NS}}$ and 
the dual Verma module $M^{\Lambda,[1],*}_{\mathrm{F}\oplus\mathrm{NS}}$ 
with the irregular vector $\langle P|=\langle 1|\otimes \langle \Lambda|$ 
parametrized as in \eqref{eq_Lambda_P} 
is defined as 
\begin{equation*}
	\langle P| Y \cdot |\Delta_0\rangle=\langle P|\cdot  Y|\Delta_0\rangle,   \quad 
  \langle P| \cdot |\Delta_0\rangle=1, 
\end{equation*}
where $Y$ is one of the generators. We put 
$\langle P| G_{-1/2}|\Delta_0\rangle=1$. 
The highest weight vector $|\Delta_0\rangle$ is the highest weight vector 
of $\Vir\oplus\Vir$ with the highest weights
\begin{equation*}
	\Delta_0^{(1)}=\frac{b^{-1}}{b^{-1}-b}\Delta_0,\quad 
	\Delta_0^{(2)}=\frac{b}{b-b^{-1}}\Delta_0,
\end{equation*}
and 
the irregular vector $\langle P|$ 
is the dual irregular vector of $\Vir\oplus\Vir$ 
with the weights
\begin{align*}
	\Lambda_1^{(1)}=\frac{b^{-1}}{b^{-1}-b}\Lambda_1,\quad 
	\Lambda_2^{(1)}=\frac{b^{-1}}{b^{-1}-b}\Lambda_2,
	\quad  
	\Lambda_1^{(2)}=\frac{b}{b-b^{-1}}\Lambda_1,\quad 
	\Lambda_2^{(2)}=\frac{b}{b-b^{-1}}\Lambda_2. 
\end{align*}
Using this pairing, we define an
\(\mathrm F\oplus\mathrm{NS}\) conformal block by 
\begin{equation*}
	\left(\langle P| 1\otimes \Phi_{\Lambda,\Lambda'}^\Delta(z)\right) \cdot|\Delta_0\rangle, 
\end{equation*}
where $\Phi_{\Lambda,\Lambda'}^\Delta(z)$ is the dual irregular vertex operator mapping $M^{\Lambda,[1],*}_{\mathrm{F}\oplus\mathrm{NS}}$ to $M^{{\Lambda'},[1],*}_{\mathrm{F}\oplus\mathrm{NS}}$. 
Put 
\begin{equation*}
    \rho_m^{\mathrm V}=\langle P,m|\cdot |\Delta_0\rangle, 
\end{equation*}
where $\langle P,m|$ is the dual irregular vector of $|P,m\rangle$. 
By Theorem~\ref{thm_IVO_decomposition}, 
\begin{align*}
	&\left(\langle P| 1\otimes \Phi_{\Lambda,\Lambda'}^\Delta(z)\right)\cdot |\Delta_0\rangle
	\\& =\sum_{2m\in\Z}p_m^{\mathrm V}(P,\Delta,\Delta_0,\beta,b)
	\langle \Lambda^{(0,1)}|V_{{\Lambda}^{(0,1)},{\Lambda'}^{(m,1)}}^{\Delta^{(1)}}(z)
	\cdot |\Delta_0^{(1)}\rangle
	\langle \Lambda^{(0,2)}|V_{{\Lambda}^{(0,2)},{\Lambda'}^{(m,2)}}^{\Delta^{(2)}}(z)
	\cdot |\Delta_0^{(2)}\rangle,
\end{align*}
where $p_m^{\mathrm V}(P,\Delta,\Delta_0,\beta,b)= \rho_m^{\mathrm V} \times  \mathsf a_{0m}^{\vee}(P,\beta,\Delta,b)$ is a scalar depending on $P$, 
$\Delta$, $\Delta_0$, $\beta$ and $b$,  
\begin{align*}
	{\Lambda'}_1^{(m,1)}=\frac{b^{-1}}{b^{-1}-b}
	\left(\Lambda_1-(\beta+2mibP_1)\right),\quad 
	{\Lambda'}_1^{(m,2)}=\frac{b}{b-b^{-1}}
	\left(\Lambda_1-(\beta+2mib^{-1}P_1)\right), 
\end{align*}
and for $i=1,2$, $V_{{\Lambda}^{(0,i)},{\Lambda'}^{(m,i)}}^{\Delta^{(i)}}(z)$ denotes the corresponding Virasoro dual irregular vertex operator. 
For $i=1,2$, we normalize
\[
 \langle {\Lambda'}^{(m,i)}|\mathbin{\cdot}|\Delta_0^{(i)}\rangle=1.
\]
Here \(\mathsf a_{0m}^{\vee}\) denotes the coefficient of the dual
component \(0\to m\); it corresponds to the primal coefficient
\(\mathsf a_{m0}\) after reversing the source and target data.

\subsection{A pairing for \texorpdfstring{$(0,2)$}{(0,2)}}
Let $L^{0}_{\mathrm{F}\oplus\mathrm{NS}}$ be the tensor product 
of the Verma module $M_\mathrm{F}$ and the irreducible Neveu--Schwarz module 
$L_{\mathrm{NS}}^{0}$ with the highest weight vector 
$ | 0\rangle$. The generators of the Neveu--Schwarz algebra 
act on $  | 0\rangle$ as 
\begin{equation*}
	L_n|0\rangle=0\quad (n\geq -1),\quad G_r| 0\rangle=0\quad (r> -1). 
\end{equation*} 
A pairing between $L^{0}_{\mathrm{F}\oplus\mathrm{NS}}$ and the dual Verma module 
$M^{\Lambda,[2],*}_{\mathrm{F}\oplus\mathrm{NS}}$ with the irregular vector 
$\langle P|=\langle 1|\otimes \langle \Lambda|$ parametrized 
as in \eqref{eq_Lambda_P} is defined as 
\begin{equation*}
  \langle P| Y \cdot |0\rangle=\langle P|\cdot  Y|0\rangle,   \quad 
  \langle P| \cdot |0\rangle=1, 
\end{equation*}
where $Y$ is one of the generators. We use the same right-module convention. We put 
$\langle P| G_{-3/2} |0\rangle=1$. The highest weight vector $|0\rangle$ is the highest weight vector 
of $\Vir\oplus\Vir$ with the highest weight $0$, and the generators of the 
Virasoro algebra also act as 
\begin{equation*}
	L_n^{(i)}|0\rangle=0\quad (n\geq -1, i=1,2). 
\end{equation*}
As in the case of $(0,0,1)$, the irregular vector $\langle P|$ is the dual 
irregular vector of $\Vir\oplus\Vir$. We define a conformal block in this case as 
\begin{equation*}
	\left(\langle P| 1\otimes \Phi_{\Lambda,\Lambda'}^\Delta(z)\right) \cdot|0\rangle, 
\end{equation*}
where $\Phi_{\Lambda,\Lambda'}^\Delta(z)$ is the dual irregular vertex operator mapping $M^{\Lambda,[2],*}_{\mathrm{F}\oplus\mathrm{NS}}$ to $M^{{\Lambda'},[2],*}_{\mathrm{F}\oplus\mathrm{NS}}$. 
Put
\begin{equation*}
    \rho_m^{\mathrm{IV}}=\langle P,m|\cdot |0\rangle,
\end{equation*}
where $\langle P,m|$ is the dual irregular vector of $|P,m\rangle$. 
By Theorem~\ref{thm_IVO_decomposition}, 
\begin{align*}
	&\left(\langle P| 1\otimes \Phi_{\Lambda,\Lambda'}^\Delta(z)\right)\cdot |0\rangle
    \\
	&=\sum_{2m\in\Z}p_m^{\mathrm{IV}}(P,\Delta,\beta,b)
	\langle \Lambda^{(0,1)}|V_{{\Lambda}^{(0,1)},{\Lambda'}^{(m,1)}}^{\Delta^{(1)}}(z)
	\cdot |0\rangle
	\langle \Lambda^{(0,2)}|V_{{\Lambda}^{(0,2)},{\Lambda'}^{(m,2)}}^{\Delta^{(2)}}(z)
	\cdot|0\rangle,
\end{align*}
where $p_m^{\mathrm{IV}}(P,\Delta,\beta,b)=\rho_m^{\mathrm{IV}} \times  \mathsf a_{0m}^{\vee}(P,\beta,\Delta,b)$ 
is a scalar depending on $P$, 
$\Delta$, $\beta$ and $b$,  
\begin{align*}
	{\Lambda'}_2^{(m,1)}=\frac{b^{-1}}{b^{-1}-b}
	\left(\Lambda_2-2(\beta+mibP_2)\right),\quad 
	{\Lambda'}_2^{(m,2)}=\frac{b}{b-b^{-1}}
	\left(\Lambda_2-2(\beta+mib^{-1}P_2)\right), 
\end{align*}
and for $i=1,2$, $V_{{\Lambda}^{(0,i)},{\Lambda'}^{(m,i)}}^{\Delta^{(i)}}(z)$ denotes the corresponding Virasoro dual irregular vertex operator. 
For $i=1,2$, we normalize
\[
 \langle {\Lambda'}^{(m,i)}|\mathbin{\cdot}|0\rangle=1.
\]

\subsection{Bilinear equations for irregular conformal blocks}

Let $H_n=bL_n^{(1)}+b^{-1}L_n^{(2)}$. Then 
\begin{equation*}
	H_n=Q\sum_{r}r:f_{n-r}f_r:-\sum_r f_{n-r}G_r. 
\end{equation*}
As explained in \cite{BS}, the insertions of operators $H_n$ into \(\mathrm F\oplus\mathrm{NS}\) conformal blocks yield bilinear equations for the Virasoro conformal blocks. We derive bilinear equations of type $(0,0,1)$ for $n=0$, and of type $(0,2)$ for $n=-1$. 

\subsubsection{Bilinear equations of type \texorpdfstring{$(0,0,1)$}{(0,0,1)}}
We set 
$\mathcal{F}_{k}=\langle P| H_{0}^k \left(1\otimes\Phi^{\Delta}_{\Lambda,\Lambda'}(z)\right)\cdot |\Delta_0\rangle$.  
\begin{prop}\label{prop_100}
	We have 
\begin{align*}
	&\mathcal{F}_1=0,
	\\
	&\mathcal{F}_2=-\langle P | G_{-1/2}G_{1/2}\left(1\otimes \Phi^{\Delta}_{\Lambda,\Lambda'}(z)\right)|\Delta_0\rangle,
	\\
	&\mathcal{F}_3=Q\mathcal{F}_2,
	\\
	&\mathcal{F}_4=\left(Q^2+1-2\Delta-2\Delta_0-2\Lambda_{1}z-2\Lambda_{2}z^2
	-2z\frac{d}{dz}\right)
	\mathcal{F}_2
	-z\left(2\Delta\Lambda_{1}+(\Lambda_{1}+2z\Lambda_{2})z\frac{d}{dz}\right)
  \mathcal{F}_0.
\end{align*}
\end{prop}
\begin{proof}
Suppress the identity operator on the free-fermion factor and put
\[
 D=z\frac{d}{dz},\qquad
 v=\Phi^{\Delta}_{\Lambda,\Lambda'}(z)|\Delta_0\rangle,
 \qquad
 \psi=\Psi^{\Delta}_{\Lambda,\Lambda'}(z)|\Delta_0\rangle.
\]
We first separate the free-fermion calculation from the NS calculation.  Set
\[
 J_n=\sum_{r\in\mathbb Z+\frac12}r:f_{n-r}f_r:,
 \qquad
 K_n=-\sum_{r\in\mathbb Z+\frac12}f_{n-r}G_r,
 \qquad H_n=QJ_n+K_n,
\]
and let $\operatorname{pr}_{\mathrm F}$ denote the projection onto the
free-fermion vacuum line $\mathbb C|1\rangle$.  Since
\[
 J_0|1\rangle=0,
 \qquad [J_0,f_{-r}]=2r f_{-r}\quad (r>0),
\]
only the free-fermion vacuum component has to be retained.  A single
contraction gives
\begin{align}
 \operatorname{pr}_{\mathrm F}H_0^2(|1\rangle\otimes w)
 &=-|1\rangle\otimes\sum_{r>0}G_{-r}G_rw,                                      \label{eq:H0-square-vac}\\
 \operatorname{pr}_{\mathrm F}H_0^3(|1\rangle\otimes w)
 &=-Q|1\rangle\otimes\sum_{r>0}2rG_{-r}G_rw.                                  \label{eq:H0-cube-vac}
\end{align}
For the fourth power, the terms containing two $J_0$'s and two $K_0$'s
give the first sum below, while the three complete contraction patterns of four
$K_0$'s give the double sum:
\begin{align}
 \operatorname{pr}_{\mathrm F}H_0^4(|1\rangle\otimes w)
 =&-Q^2|1\rangle\otimes\sum_{r>0}(2r)^2G_{-r}G_rw                              \notag\\
 &+|1\rangle\otimes\sum_{r,s>0}
 \Bigl(
   G_{-r}G_rG_{-s}G_s
  -G_{-r}G_{-s}G_rG_s
  +G_{-r}G_{-s}G_sG_r
 \Bigr)w.                                                                       \label{eq:H0-fourth-vac}
\end{align}
Here and below the summation indices are half-integers.  These formulas
contain all the required free-fermion calculations; in particular, no full
expansion of $H_0^k$ is needed.

The rank-one dual irregular-vector relations are
\[
 \langle P|G_{-r}=0\quad(r>1),\qquad
 \langle P|L_{-1}=\Lambda_1\langle P|,
 \qquad
 \langle P|L_{-2}=\Lambda_2\langle P|,
\]
with $\langle P|L_{-n}=0$ for $n>2$.  Consequently, in
\eqref{eq:H0-square-vac}--\eqref{eq:H0-fourth-vac}, the outer summation
index $r$ can only be $1/2$.  The first three assertions follow at once:
\[
 \mathcal F_1=0,
 \qquad
 \mathcal F_2=-\langle P|G_{-1/2}G_{1/2}v\rangle,
 \qquad
 \mathcal F_3=Q\mathcal F_2.
\]

It remains to evaluate the four-$K_0$ contribution.  Put 
\[
 T_s=
 G_{-1/2}G_{1/2}G_{-s}G_s
 -G_{-1/2}G_{-s}G_{1/2}G_s
 +G_{-1/2}G_{-s}G_sG_{1/2},
 \qquad s>0.
\]
The intertwining relations and the highest-weight conditions give
\[
 G_rv=z^{r+\frac12}\psi,
 \qquad
 G_r\psi=z^{r-\frac12}\bigl(D+(2r+1)\Delta\bigr)v
 \qquad(r>0).
\]
Using these formulas together with
$\{G_{1/2},G_{-s}\}=2L_{1/2-s}$, one obtains, without expanding any other
modes,
\begin{equation}\label{eq:Ts-rank-one}
 T_sv=z^{s+\frac12}G_{-1/2}
 \left
 \{
   G_{-s}\bigl(-D+(2s-3)\Delta\bigr)v
   +2L_{1/2-s}\psi
 \right
 \}.
\end{equation}
This identity is the main simplification: after pairing with $\langle P|$,
the annihilation conditions imply that only $s=1/2,3/2,5/2$ contribute.  Indeed,
\begin{align*}
 &\langle P|G_{-1/2}G_{-1/2}=\Lambda_1\langle P|,
 \qquad
 \langle P|G_{-1/2}G_{-3/2}=2\Lambda_2\langle P|,\\
 &\langle P|G_{-1/2}L_{-1}=\Lambda_1\langle P|G_{-1/2},
 \qquad
 \langle P|G_{-1/2}L_{-2}=\Lambda_2\langle P|G_{-1/2}. 
\end{align*}
Since $G_{1/2}v=z\psi$, we have
$z\langle P|G_{-1/2}\psi\rangle=-\mathcal F_2$.  Moreover,
\[
 L_0\psi=\left(D+\Delta+\frac12+\Delta_0\right)\psi.
\]
Substitution into \eqref{eq:Ts-rank-one} gives
\begin{align*}
 \langle P|T_{1/2}v
 &=\bigl(1-2D-2\Delta-2\Delta_0\bigr)\mathcal F_2
   -z\Lambda_1(D+2\Delta)\mathcal F_0,\\
 \langle P|T_{3/2}v
 &=-2\Lambda_1z\mathcal F_2-2\Lambda_2z^2D\mathcal F_0,\\
 \langle P|T_{\frac52}v
 &=-2\Lambda_2z^2\mathcal F_2,
\end{align*}
while $\langle P|T_sv\rangle=0$ for $s\geq 7/2$.  Finally, the first
term of \eqref{eq:H0-fourth-vac} contributes $Q^2\mathcal F_2$.
Summing the three displayed contributions 
yields
\begin{equation*}
 \mathcal F_4
 =\left(Q^2+1-2\Delta-2\Delta_0
 -2\Lambda_1z-2\Lambda_2z^2-2D\right)\mathcal F_2
 -z\left(2\Delta\Lambda_1
 +(\Lambda_1+2z\Lambda_2)D\right)\mathcal F_0,
\end{equation*}
which is the required fourth relation.
\end{proof}
Let the bilinear operators $\mathcal{D}_b^k$ ($k\geq 0$) be defined by
\begin{equation*}
	\mathcal{D}_b^k(f\cdot g)=\sum_{i=0}^k b^{2i-k}\binom{k}{i}
	\left(z\frac{d}{dz}\right)^if\, 
\left(z\frac{d}{dz}\right)^{k-i}g. 
\end{equation*}
For example, 
\begin{align*}
&\mathcal{D}_b^0(f\cdot g)=fg,
\\
    &\mathcal{D}_b^1(f\cdot g)=b\left(z\frac{df}{dz}\right)g+b^{-1}f\left(z\frac{dg}{dz}\right),\\ 
    &\mathcal{D}_b^2(f\cdot g)= b^2 \left(z\frac{d}{dz}\right)^2 f \cdot g + 2 \left(z\frac{df}{dz} \right) \left(z\frac{dg}{dz} \right) + b^{-2} f \cdot \left(z\frac{d}{dz}\right)^2 g.
\end{align*}

By the definition of $H_0$, we obtain  
\begin{equation*}
 \mathcal F_k
 =\sum_{2m\in\mathbb Z}p_m^{\mathrm V}
 \mathcal D_b^k\Bigl(
 \langle \Lambda^{(0,1)}|
 V_{{\Lambda}^{(0,1)},{\Lambda'}^{(m,1)}}^{\Delta^{(1)}}(z)
 \mathbin{\cdot}|\Delta_0^{(1)}\rangle
 \mathbin{\cdot}
 \langle \Lambda^{(0,2)}|
 V_{{\Lambda}^{(0,2)},{\Lambda'}^{(m,2)}}^{\Delta^{(2)}}(z)
 \mathbin{\cdot}|\Delta_0^{(2)}\rangle
 \Bigr).
\end{equation*}
Hence, the relations in Proposition~\ref{prop_100} can be viewed as  
bilinear equations for weighted sums of products of Virasoro irregular conformal blocks of type $(0,0,1)$.  

Put  
\begin{align*}
D_b^{1,\mathrm{V}}=&\mathcal D_b^1,\quad D_b^{3,\mathrm{V}}=\mathcal D_b^3-Q\mathcal D_b^2,
\\
	D_b^{4,\mathrm{V}}=&\mathcal D_b^4-\left(Q^2+1-2\Delta-2\Delta_0-2\Lambda_{1}z-2\Lambda_{2}z^2
	-2z\frac{d}{dz}\right)
	\mathcal D_b^2
	\\
    &{} +z\left(2\Delta\Lambda_{1}+(\Lambda_{1}+2z\Lambda_{2})z\frac{d}{dz}\right)\mathcal D_b^0. 
\end{align*}

\subsubsection{Bilinear equations of type \texorpdfstring{$(0,2)$}{(0,2)}}
Since \(\Lambda_3\) can be set to zero by a gauge transformation,
we assume \(\Lambda_3=0\) throughout this subsection. We set
\[
 \mathcal{F}_{k}
 =
 \langle P|H_{-1}^k
 \left(1\otimes\Phi^{\Delta}_{\Lambda,\Lambda'}(z)\right)
 \cdot|0\rangle.
\]
A direct calculation gives the following identities. 
\begin{prop}\label{prop_20}
	We have
	\begin{align*}
		&\mathcal{F}_{1}=0,
		\\
		&\mathcal{F}_{2}=-\langle P| G_{-3/2}G_{-1/2}
		\left(1\otimes\Phi^{\Delta}_{\Lambda,\Lambda'}(z)\right)|0\rangle,
		\\
		&\mathcal{F}_{3}=0,
		\\
		&\mathcal{F}_{4}=-2\left(\Lambda_{2}+\Lambda_{4}z^2\right)\mathcal{F}_{2}
		+\left(4\Delta\Lambda_{4}-2\Lambda_{4}z\frac{d}{dz}\right)\mathcal{F}_{0}. 
	\end{align*}
	\end{prop}
\begin{proof}
We use the same free-fermion vacuum projection as in the preceding proof.
With $J_n$ and $K_n$ defined there, one has
\[
 J_{-1}|1\rangle=0,
 \qquad
 [J_{-1},f_s]=(1-2s)f_{s-1}.
\]
Therefore, for an arbitrary NS vector $w$,
\begin{align}
 \operatorname{pr}_{\mathrm F}H_{-1}^2(|1\rangle\otimes w)
 =&-|1\rangle\otimes\sum_{r>-1}G_{-2-r}G_rw,                                  \label{eq:Hm1-square-vac}\\
 \operatorname{pr}_{\mathrm F}H_{-1}^3(|1\rangle\otimes w)
 =&-Q|1\rangle\otimes\sum_{r>-1}(2r+3)G_{-3-r}G_rw,                           \label{eq:Hm1-cube-vac}\\
 \operatorname{pr}_{\mathrm F}H_{-1}^4(|1\rangle\otimes w)
 =&-Q^2|1\rangle\otimes
   \sum_{r>-1}(2r+3)(2r+5)G_{-4-r}G_rw                                        \notag\\
 &+|1\rangle\otimes\sum_{r,s>-1}
 \begin{aligned}[t]
 \Bigl(&G_{-2-r}G_rG_{-2-s}G_s\\
       &-G_{-2-r}G_{-2-s}G_rG_s\\
       &+G_{-2-r}G_{-2-s}G_sG_r\Bigr)w.
 \end{aligned}                                                                  \label{eq:Hm1-fourth-vac}
\end{align}
Again, the last line is simply the sum of the three complete contraction patterns;
the factors $2r+3$ and $2r+5$ come from the two successive applications
of $[J_{-1},f_s]=(1-2s)f_{s-1}$.

For the rank-two dual irregular vector,
\[
 \langle P|G_{-r}=0\quad(r>2),\qquad
 \langle P|L_{-n}=\Lambda_n\langle P|
 \quad(n=2,3,4),
\]
and $\langle P|L_{-n}=0$ for $n>4$.  Hence
\eqref{eq:Hm1-square-vac} has only the term $r=-1/2$, whereas every term
in \eqref{eq:Hm1-cube-vac} is annihilated by $\langle P|$.  The same is
true of the $Q^2$-term in \eqref{eq:Hm1-fourth-vac}.  Thus
\[
 \mathcal F_1=0,
 \qquad
 \mathcal F_2=-\langle P|G_{-3/2}G_{-1/2}v\rangle,
 \qquad
 \mathcal F_3=0,
\]
where, suppressing the free-fermion factor,
\[
 D=z\frac{d}{dz},\qquad
 v=\Phi^{\Delta}_{\Lambda,\Lambda'}(z)|0\rangle,
 \qquad
 \psi=\Psi^{\Delta}_{\Lambda,\Lambda'}(z)|0\rangle.
\]

It remains to consider the four-$K_{-1}$ term.  In its outer sum, only
$r=-1/2$ survives.  For $s>-1$, put
\begin{align*}
 T_s={}&G_{-3/2}G_{-1/2}G_{-2-s}G_s
 -G_{-3/2}G_{-2-s}G_{-1/2}G_s
 +G_{-3/2}G_{-2-s}G_sG_{-1/2}.
\end{align*}
Since $G_r|0\rangle=0$ for $r>-1$, the intertwining relations imply
\[
 G_rv=z^{r+\frac12}\psi,
 \qquad
 G_r\psi=z^{r-\frac12}\bigl(D+(2r+1)\Delta\bigr)v
 \qquad(r>-1).
\]
Using these relations and
$\{G_{-1/2},G_{-2-s}\}=2L_{-5/2-s}$, we obtain
\begin{equation}\label{eq:Ts-rank-two}
 T_sv=
 z^{s-\frac12}G_{-3/2}G_{-2-s}
 \bigl(-D+(2s+1)\Delta\bigr)v
 +2z^{s+\frac12}G_{-3/2}L_{-5/2-s}\psi.
\end{equation}
The rank-two conditions now reduce the sum to three values of $s$.  More
precisely, using $\Lambda_3=0$,
\begin{align*}
 \langle P|G_{-3/2}G_{-2-s}
 &=
 \begin{cases}
  0,&s=-\frac12,\\
  2\Lambda_4\langle P|,&s=\frac12,\\
  0,&s\geq\frac32,
 \end{cases}\\
 \langle P|G_{-3/2}L_{-\frac52-s}
 &=
 \begin{cases}
  \Lambda_2\langle P|G_{-3/2},&s=-\frac12,\\
  0,&s=\frac12,\\
  \Lambda_4\langle P|G_{-3/2},&s=\frac32,\\
  0,&s\geq\frac52.
 \end{cases}
\end{align*}
The commutator terms produced when moving $L_{-2}$ or $L_{-4}$ to the
left contain $G_r$ with $r<-2$ and hence vanish against $\langle P|$.
Since $G_{-1/2}v=\psi$, formula \eqref{eq:Ts-rank-two} gives
\begin{align*}
 \langle P|T_{-1/2}v&=-2\Lambda_2\mathcal F_2,\\
 \langle P|T_{1/2}v
 &=2\Lambda_4(-D+2\Delta)\mathcal F_0,\\
 \langle P|T_{3/2}v&=-2\Lambda_4z^2\mathcal F_2,
\end{align*}
with all remaining terms equal to zero.  Summing the three displayed contributions yields 
\[
 \mathcal F_4
 =-2(\Lambda_2+\Lambda_4z^2)\mathcal F_2
 +(4\Delta\Lambda_4-2\Lambda_4D)\mathcal F_0.
\]
\end{proof}

	Let the bilinear operators $D_b^k$ ($k\geq 0$) be defined by 
	\begin{equation*}
		D_b^k(f\cdot g)=\sum_{i=0}^k b^{2i-k}\binom{k}{i}
		\left(\frac{d}{dz}\right)^if\, 
	\left(\frac{d}{dz}\right)^{k-i}g. 
	\end{equation*}	
	Then, by the definition of $H_{-1}$, we get the following. 
\begin{align*}
 \mathcal F_k
 =&\sum_{2m\in\mathbb Z}p_m^{\mathrm{IV}}
  D_b^k\Bigl(
 \langle \Lambda^{(0,1)}|
 V_{{\Lambda}^{(0,1)},{\Lambda'}^{(m,1)}}^{\Delta^{(1)}}(z)
 \mathbin{\cdot}|0\rangle
 \mathbin{\cdot}
 \langle \Lambda^{(0,2)}|
 V_{{\Lambda}^{(0,2)},{\Lambda'}^{(m,2)}}^{\Delta^{(2)}}(z)
 \mathbin{\cdot}|0\rangle
 \Bigr).
\end{align*}
Hence, the relations in Proposition~\ref{prop_20} can be viewed as  
bilinear equations for weighted sums of products of Virasoro irregular conformal blocks of type $(0,2)$. 

Put  
\begin{align*}
&D_b^{1,\mathrm{IV}}= D_b^1,\quad D_b^{3,\mathrm{IV}}= D_b^3,
\\
&D_b^{4,\mathrm{IV}}=D_b^4+2\left(\Lambda_{2}+z^2\Lambda_{4}\right)D_b^2
	-2\Lambda_4\left(2\Delta-z\frac{d}{dz}\right) D_b^0. 
\end{align*}

\subsection{Comparison with quantum Painlev\'e bilinear operators}

Let $\epsilon_1,\epsilon_2\in\C$ and $\epsilon=\epsilon_1+\epsilon_2$. 
The quantum Painlev\'e tau function is defined by the Zak transform of a function $Z(a_D; \epsilon_1, \epsilon_2 | t)$ as 
\begin{align*}
\tau(a_D, \eta; \epsilon_1, \epsilon_2 | t) = \sum_{n \in \mathbb{Z}} e^{in\eta} Z(a_D + n\epsilon_2; \epsilon_1, \epsilon_2 | t),
\end{align*}
where the variables $(a_D, \eta)$ are assumed to be non-commutative 
and satisfy 
\begin{equation*}
    a_D e^{i\eta} = e^{i\eta} (a_D + \epsilon).
\end{equation*}
This non-commutativity was first used in \cite{BGM} to consider the quantum $q$-Painlevé tau function, and later in \cite{BST} for quantum Painlev\'e tau functions.

The bilinear equations for the quantum fifth Painlev\'e tau function in \cite{BST} are 
\begin{align}
& D_{\epsilon_1, \epsilon_2 [\ln t]}^1 (\tau^{(1)}, \tau^{(2)}) = 0, \label{eq:bilinear_PV_1}
\\
& D_{\epsilon_1, \epsilon_2 [\ln t]}^3 (\tau^{(1)}, \tau^{(2)}) = \epsilon D_{\epsilon_1, \epsilon_2 [\ln t]}^2 (\tau^{(1)}, \tau^{(2)}),\label{eq:bilinear_PV_3}
\\ 
     &  
     D_{\epsilon_1, \epsilon_2 [\ln t]}^4 (\tau^{(1)}, \tau^{(2)}) + 2 \left( \epsilon_1 \epsilon_2 \frac{d}{d \ln t} \right) D_{\epsilon_1, \epsilon_2 [\ln t]}^2 (\tau^{(1)}, \tau^{(2)}) 
     \notag \\ 
     & \quad 
     - \left(\frac{t^2}{4}+(e_1^{[3]}+\epsilon)t+\epsilon_1\epsilon_2 + \epsilon^2\right) D_{\epsilon_1, \epsilon_2 [\ln t]}^2 (\tau^{(1)}, \tau^{(2)})\notag \\
     &\quad -\frac{t}{4}\left(t+2(e_1^{[3]}+\epsilon)\right) \left(\epsilon_1\epsilon_2 \frac{d}{d\ln t}(\tau^{(1)}\tau^{(2)})\right) 
      + \frac{t}{8}\left((e_2^{[3]}+e_1^{[3]}\epsilon+\epsilon^2)t + 2e_3^{[3]}\right)\tau^{(1)}\tau^{(2)} = 0, 
     \label{eq:bilinear_PV_4}
\end{align}
where $e_i^{[3]}$ ($i=1,2,3$) are complex parameters. 
Here, the generalized Hirota differential operators $D_{\epsilon_1, \epsilon_2 [\ln t]}^k$ are defined by
\begin{equation*}
    D^{k}_{\epsilon_1, \epsilon_2 [\ln t]}(f, g)=
    \sum_{i=0}^k \epsilon_1^i\epsilon_2^{k-i}\binom{k}{i}
	\left(t\frac{d}{dt}\right)^if\, 
\left(t\frac{d}{dt}\right)^{k-i}g, 
\end{equation*}
and the shifted $\tau$-functions are defined by
\begin{equation*}
\tau^{(1)}(a_D, \eta; \epsilon_1, \epsilon_2 | t) 
= \tau(a_D, \eta; 2\epsilon_1, \epsilon_2 - \epsilon_1 | t), \quad 
\tau^{(2)}(a_D, \eta; \epsilon_1, \epsilon_2 | t)
= \tau(a_D, \eta; \epsilon_1 - \epsilon_2, 2\epsilon_2 | t). 
\end{equation*}

The bilinear equations for the quantum fourth Painlev\'e tau function in \cite{BST} are 
\begin{align}
    & D_{\epsilon_1, \epsilon_2 [ t]}^1 (\tau^{(1)}, \tau^{(2)}) = 0, \label{eq:bilinear_PIV_1}
\\
& D_{\epsilon_1, \epsilon_2 [ t]}^3 (\tau^{(1)}, \tau^{(2)}) = 0,\label{eq:bilinear_PIV_3}
\\ 
     &  
     D_{\epsilon_1,\epsilon_2 [t]}^4 (\tau^{(1)}, \tau^{(2)}) - \left( \frac{t^2}{4} - 6(m_1+m_2) - \frac{3}{2}\epsilon \right) D_{\epsilon_1,\epsilon_2 [t]}^2 (\tau^{(1)}, \tau^{(2)}) \notag\\
     &\quad - \frac{1}{4}t \left( \epsilon_1\epsilon_2 \frac{d}{dt} \right) (\tau^{(1)}\tau^{(2)}) + \left( 2m_1+m_2+\frac{\epsilon}{2} \right)\left( m_1+2m_2+\frac{\epsilon}{2} \right)\tau^{(1)}\tau^{(2)} = 0,  
     \label{eq:bilinear_PIV_4}
\end{align}
where $m_1,m_2$ are complex parameters.
Here, the bilinear operators $D^k_{\epsilon_1,\epsilon_2}$ ($k\geq 1$) are defined by
\begin{equation*}
	D_{\epsilon_1,\epsilon_2}^k(f\cdot g)=\sum_{i=0}^k \epsilon_1^i\epsilon_2^{k-i}\binom{k}{i}
	\left(\frac{d}{dt}\right)^if\, 
\left(\frac{d}{dt}\right)^{k-i}g. 
\end{equation*}
\begin{re}
    When \(\epsilon=0\), the equations \eqref{eq:bilinear_PV_1},  \eqref{eq:bilinear_PV_3},  \eqref{eq:bilinear_PIV_1}, \eqref{eq:bilinear_PIV_3} are identically satisfied, and \eqref{eq:bilinear_PV_4}, \eqref{eq:bilinear_PIV_4} are equivalent to the well-known bilinear equations for the fifth and fourth Painlev\'e tau functions \cite{Okamoto80}, respectively. In the case of \(\epsilon\neq 0\), it is observed that the three bilinear equations determine the quantum Painlev\'e tau functions. 
\end{re}

We assume
\[
 \epsilon_1\epsilon_2\neq0,
 \qquad
 \epsilon_1\neq\epsilon_2.
\]
We use the following parametrization:
\[
 z=t,
 \qquad
 Q=b+b^{-1},
 \qquad
 b=\frac{\epsilon_1}{\sqrt{\epsilon_1\epsilon_2}},
 \qquad
 b^{-1}=\frac{\epsilon_2}{\sqrt{\epsilon_1\epsilon_2}}.
\]
For the quantum $\mathrm{P}_{\mathrm{V}}$ tau-function equations, assume in addition that
$e_1^{[3]}+\epsilon\neq0$.  Apply the gauge transformation
\[
 \tau^{(i)}(t)=t^{\gamma_i}F^{(i)}(t).
\]
After division by the common nonzero factor
\((\epsilon_1\epsilon_2)^{j/2}\), the bilinear operators for $F^{(i)}(t)$
coincide with $D_b^{j,\mathrm{V}}$, $j=1,3,4$, under the following
parameter identification:
\begin{align*}
    &\Lambda_2=-\frac{1}{8\epsilon_1\epsilon_2}, \quad \Lambda_1=-\frac{e_1^{[3]}+\epsilon}{2\epsilon_1\epsilon_2},    
    \\
&\Delta+\Delta_0=\frac{e_2^{[3]}+e_1^{[3]}\epsilon+\epsilon^2}{2\epsilon_1\epsilon_2},\quad 
\Delta_0-\Delta=\frac{e_3^{[3]}}{2\epsilon_1\epsilon_2(e_1^{[3]}+\epsilon)},
    \\
    &\gamma_1
=
\frac{e_2^{[3]}+e_1^{[3]}\epsilon+\epsilon^2}
{2\epsilon_1(\epsilon_2-\epsilon_1)},
\qquad
\gamma_2
=
-\frac{e_2^{[3]}+e_1^{[3]}\epsilon+\epsilon^2}
{2\epsilon_2(\epsilon_2-\epsilon_1)}. 
    \end{align*}
For the quantum $\mathrm{P}_{\mathrm{IV}}$ tau-function equations, after division by the
common nonzero factor \((\epsilon_1\epsilon_2)^{j/2}\), the operators
appearing in \eqref{eq:bilinear_PIV_1}, \eqref{eq:bilinear_PIV_3}, and
\eqref{eq:bilinear_PIV_4} coincide with $D_b^{j,\mathrm{IV}}$,
$j=1,3,4$, under the following parameter identification:
\begin{align*}
    &\Lambda_4=-\frac{1}{8\epsilon_1\epsilon_2},\quad 
    \Lambda_2=\frac{1}{2\epsilon_1\epsilon_2}\left(6(m_1+m_2) +\frac{3}{2}\epsilon\right), 
    \\
    &\Delta=\frac{2}{\epsilon_1\epsilon_2}\left( 2m_1+m_2+\frac{\epsilon}{2} \right)\left( m_1+2m_2+\frac{\epsilon}{2} \right). 
\end{align*}

\section{Discussion}

Recently, the existence and uniqueness of genuinely irregular vectors for the
Virasoro algebra of arbitrary rank were proved in \cite{Nagoya26}. We expect
that the algebraic framework developed there extends to super-Virasoro
algebras and \(W\)-algebras. In particular, in the Neveu--Schwarz case, the
same framework should lead to bilinear equations for the quantum
\(\mathrm{P}_{\mathrm{II}}\) and
\(\mathrm{P}_{\mathrm{I}}\) tau functions.

In the present paper, we have established the coincidence between the
bilinear differential operators obtained from the decomposition of
Neveu--Schwarz irregular vertex operators and those appearing in the quantum
Painlev\'e tau-function 
equations. A stronger result would be an explicit identification
of the quantum Painlev\'e tau functions with Zak transforms of Virasoro
irregular conformal blocks. For such an identification, one needs explicit
formulas for the coefficients \(p_{mn}\) occurring in the Zak-transform
expansion. These coefficients have the factorized form
\[
p_{mn}=\mathsf a_{mn}\rho_{mn},
\]
where \(\mathsf a_{mn}\) is the scalar coefficient in the decomposition of the
irregular vertex operator, and \(\rho_{mn}\) is the pairing of the Virasoro
irregular vectors embedded in the Neveu--Schwarz and free-fermion Verma
modules. In Section~5, the coefficients are
\[
 p_m^{\mathrm V}=\mathsf a_{0m}^{\vee}\rho_m^{\mathrm V},
 \qquad
 p_m^{\mathrm{IV}}=\mathsf a_{0m}^{\vee}\rho_m^{\mathrm{IV}}.
\]
The pairing factors \(\rho_{mn}\) should be computable explicitly using
properties of Kummer and Hermite functions, whereas the coefficients
\(\mathsf a_{mn}\) can be determined by degenerating the decomposition formula for
regular vertex operators. Combining these two calculations would give the
explicit coefficients \(p_{mn}\), and hence the Zak-transform formulas for
the corresponding quantum Painlev\'e tau functions. We do not carry out this
analysis here, since the main purpose of the present paper is the construction
and decomposition of irregular Neveu--Schwarz vertex operators and the
representation-theoretic derivation of the bilinear operators.

It is natural to expect that the bilinear equations for quantum Painlev\'e
tau functions characterize the tau functions themselves. A difficulty in
proving this directly is that these equations form highly overdetermined
systems. Although degeneration procedures give proofs of the Zak-transform
representations of quantum tau functions in several cases, it remains
important to obtain a direct proof that does not rely on limiting procedures
or explicit special-function formulas. The fact that the weighted sums of products of irregular conformal blocks constructed in this paper satisfy the corresponding bilinear equations 
should provide a useful starting point for proving directly that the quantum
Painlev\'e bilinear equations uniquely determine the tau functions as Zak
transforms of irregular conformal blocks.

\section*{Statements and Declarations}

\paragraph{Funding}
This work was supported by the Japan Society for the Promotion of
Science (JSPS) through KAKENHI Grant Number 22K03350.

\paragraph{Competing interests}
The author declares no competing interests.

\paragraph{Data availability}
No datasets were generated or analyzed during the current study.

\end{document}